\documentclass[draftclsnofoot, onecolumn, 12pt]{IEEEtranTCOM}

\usepackage{graphicx}
\usepackage{epsfig}
\usepackage[usenames,dvipsnames,svgnames,table]{xcolor}
\usepackage{subfigure}
\usepackage{scalefnt}
\usepackage{amsbsy}
\usepackage{algorithm,algorithmic}
\usepackage{amsmath,amssymb,amsfonts,cite,hhline}

\DeclareMathOperator*{\argmax}{argmax}
\DeclareMathOperator*{\argmin}{argmin}

\def\bbC{{\mathbb{C}}}

\def\cA{{\mathcal{A}}}

\def\cC{{\mathcal{C}}}

\def\cF{{\mathcal{F}}}

\def\cN{{\mathcal{N}}}

\def\cP{{\mathcal{P}}}

\def\cS{{\mathcal{S}}}
\def\cT{{\mathcal{T}}}

\def\cZ{{\mathcal{Z}}}

\def\argmin{\mathop{\mathrm{argmin}}}
\def\argmax{\mathop{\mathrm{argmax}}}

\def\b{{\mathbf{\Sigma}}}

\def\b0{{\mathbf{0}}}

\def\ba{{\mathbf{a}}}

\def\bc{{\mathbf{c}}}
\def\bd{{\mathbf{d}}}

\def\bff{{\mathbf{f}}}
\def\bg{{\mathbf{g}}}
\def\bh{{\mathbf{h}}}

\def\bu{{\mathbf{u}}}
\def\bv{{\mathbf{v}}}

\def\by{{\mathbf{y}}}
\def\bz{{\mathbf{z}}}
\def\b0{{\mathbf{0}}}

\def\bA{{\mathbf{A}}}

\def\bF{{\mathbf{F}}}

\def\bH{{\mathbf{H}}}

\def\bZ{{\mathbf{Z}}}

\newcommand{\barr}{\begin{array}}
\newcommand{\earr}{\end{array}}
\newcommand{\bmat}{\left[\begin{array}}
\newcommand{\emat}{\end{array}\right]}
\newcommand{\bequ}{\begin{equation}}
\newcommand{\eequ}{\end{equation}}

\newcommand{\bet}{\begin{table}}
\newcommand{\eet}{\end{table}}
\newcommand{\btt}{\begin{tabular}}
\newcommand{\ett}{\end{tabular}}
\newcommand{\bec}{\begin{center}}
\newcommand{\eec}{\end{center}}
\newcommand{\bef}{\begin{figure}}
\newcommand{\eef}{\end{figure}}
\newcommand{\beq}{\begin{eqnarray}}
\newcommand{\eeq}{\end{eqnarray}}
\newcommand{\beqn}{\begin{eqnarray*}}
\newcommand{\eeqn}{\end{eqnarray*}}
\newcommand{\eeqs}{\end{eqnarray}\vspace*{-0.05in}}
\newcommand{\bit}{\begin{itemize}}
\newcommand{\eit}{\end{itemize}}
\newcommand{\bed}{\begin{description}}
\newcommand{\eed}{\end{description}}
\newcommand{\ben}{\begin{enumerate}}
\newcommand{\een}{\end{enumerate}}
\newcommand{\bis}{\vspace*{-0.05in}\begin{itemize}\small}
\newcommand{\eis}{\end{itemize}\normalsize\vspace*{-0.05in}}

\def\-{\! - \!}
\def\+{\! + \!}
\def\={\! = \!}
\def\>{\! > \!}

\def\Pr{{\mathrm{Prob}}}

\def\exp{{\mathrm{exp}}}

\newenvironment{proof}[1][Proof]{\begin{trivlist}
\item[\hskip \labelsep {\bfseries #1}]}{\end{trivlist}}

\newcommand{\C}{\mathbb{C}}

\newcommand{\qed}{\nobreak \ifvmode \relax \else
      \ifdim\lastskip<1.5em \hskip-\lastskip
      \hskip1.5em plus0em minus0.5em \fi \nobreak
      \vrule height0.75em width0.5em depth0.25em\fi}

\def\C{{\mathbb C}}

                 \def\hell{\ell_{opt}}
\def\Mt{M_t}            \def\Mr{M_r}                 
\def\gam{\gamma}
\def\argmin{\mathop{\mathrm{argmin}}}
\def\argmax{\mathop{\mathrm{argmax}}}
\def\inf{\mathop{\mathrm{inf}}}
\def\lp{\left(}     \def\rp{\right)}    \def\ls{\left\{}    \def\rs{\right\}}    \def\lS{ \left[ }
\def\rS{ \right] }  \def\la{\left|}     \def\ra{\right|}         
\def\var{\mathop{\mathrm{Var}}}

\newtheorem{theorem}{Theorem}[section]
\newtheorem{lemma}[theorem]{Lemma}

\DeclareMathAlphabet{\mathpzc}{OT1}{pzc}{m}{it}

\normalsize  

\begin{document}

\title{
{Millimeter Wave Beamforming for Wireless Backhaul and Access in Small Cell Networks}}
\author{\small
Sooyoung~Hur$^{\dag}$, Taejoon~Kim$^{\ddag}$, David~J.~Love$^{\dag}$, James~V.~Krogmeier$^{\dag}$,
Timothy A. Thomas$^{\S}$, and Amitava Ghosh$^{\S}$\\
$^{\dag}$School of Electrical and Computer Engineering, Purdue University, West Lafayette, IN 47907 \\
         Email: \{shur, djlove,  jvk\}@ecn.purdue.edu\\
$^{\ddag}$Department of Electronic Engineering, City University of Hong Kong, Kowloon, Hong Kong \\
                 Email: taejokim@cityu.edu.hk\\
$^{\S}$Nokia Siemens Networks, Arlington Heights, IL 60004 \\
         Email: \{timothy.thomas, amitava.ghosh\}@nsn.com
\thanks{Parts of this paper were presented at the Globecom Workshop, Houston, TX, Dec. 5$-$9, 2011 \cite{Hur:2011}.}
}

\maketitle

\vspace{-2.0cm}
\begin{abstract}
\vspace{-0.4cm}
Recently, there has been considerable interest in new tiered network cellular architectures,  which would likely use many more cell sites than found today. Two major challenges will be i) providing backhaul to all of these cells and ii) finding efficient techniques to leverage higher frequency bands for mobile access and backhaul. This paper proposes the use of outdoor millimeter wave communications for backhaul networking between cells and mobile access within a cell.   To overcome the outdoor impairments found in millimeter wave propagation, this paper studies beamforming using large arrays.  However, such systems will require narrow beams, increasing sensitivity to movement caused by pole sway and other environmental concerns.   To overcome this, we propose an efficient beam alignment technique using adaptive subspace sampling and hierarchical beam codebooks. A wind sway analysis is presented to establish a notion of beam coherence time.  This highlights a previously unexplored tradeoff between array size and wind-induced movement.  Generally, it is not possible to use larger arrays without risking a corresponding performance loss from wind-induced beam misalignment.  The performance of the proposed alignment technique is analyzed and compared with other search and alignment methods. The results show significant performance improvement with reduced search time.
\end{abstract}

\vspace{-0.4cm}
\begin{IEEEkeywords}
\vspace{-0.5cm}
Millimeter wave, array antenna, beam alignment, beamforming codebook design, wind-induced vibration 
\end{IEEEkeywords}

\vspace{-0.5cm}
\section{Introduction}
\vspace{-0.2cm}
The exponential growth of demand for mobile multimedia services has motivated extensive research into improved spectrum efficiency  using techniques that increase geographic spectrum reusability, such as multi-tier cell deployment (e.g., picocell and femtocell networks) \cite{SmallCellForum:WhtPaper}.
Picocell networks are expected to be typically deployed to support demand from small, high throughput areas (e.g., urban centers, office buildings, shopping malls, train stations). 
One of the major impediments to deployment of heterogeneous small cell networks, such as picocell networks, is access to cost effective, reliable, and scalable backhaul networks.

These dense picocell deployments will make expensive wired backhaul infeasible  \cite{Chia:2009gs}.   It is also unrealistic to use existing cellular spectrum holdings for large-scale in-band backhaul, especially given the self-described ``spectrum crunch"  which cellular operators have recently lamented. Therefore, a scalable solution is to consider backhaul and access using carrier frequencies  outside of the traditional wireless bands.
Millimeter wave bands, the unlicensed 60 GHz band and the lower interference licensed 70 GHz to 80 GHz band, are a possible solution to the problem of providing small cell backhaul and access in tiered cellular networks.
Picocell networks could employ these millimeter wave backhaul networks using a variety of architectures including point-to-point links or backhaul aggregation  using an aggregator at the macrocell connected to a tree or mesh structured network.
The advantages of millimeter wave bands include the availability of many gigahertz of underutilized spectrum \cite{Chia:2009gs} and the line-of-sight (LOS) nature of millimeter wave communication which helps to control interference between systems.
However, millimeter wave systems require a large directional gain in order to combat their relatively high path loss compared to systems with lower frequencies and the additional losses due to rain and oxygen absorption.

To achieve this large directional gain, either a large physical aperture or a phased array antenna must be employed.
A large physical aperture is not possible due to  a very costly installation and the expected maintenance costs related to wind loading and other misalignments.
Thanks to the small wavelength of millimeter wave signals, large-sized phased-array antennas are able to offer large beamforming gain while keeping individual antenna elements small and cheap.
They also enable adaptive alignment of transmit and receive beams in order to relax cost requirements (e.g., relative to parabolic antennas)  for initial pointing accuracy and maintenance.

Beamforming techniques at millimeter wave have been widely researched in many standards including IEEE 802.15.3c (TG3c) \cite{IEEE:2009hg} for indoor wireless personal area networks (WPAN), IEEE 802.11ad (TGad) \cite{IEEE:80211ad} and Wireless Gigabit Alliance (WiGig) on wireless local area networks (WLAN), ECMA-387, and WirelessHD, which is focused on uncompressed HDTV streaming.
More specifically, beamforming techniques have been proposed for indoor office environments as applied in WPAN for ranges of a few meters \cite{Wang:2009jn}.
In the WLAN arena, a one-sided beam search using a beamforming codebook has been employed to establish the initial alignment between large-sized array antennas \cite{Cordeiro:2010if}.
However, beamforming methods used for indoor scenarios do not easily extend to outdoor scenarios where longer distances, outdoor propagation, and other environmental factors such as wind and precipitation can cause as much as a 48 dB receive SNR degradation \cite{HongZhang:2010dv} and thus require a much larger beamforming gain and more subtle beam alignment. 
In particular, since picocell units will be mounted to outdoor structures such as poles,  vibration and movement induced by wind flow and gusts have the potential to cause unacceptable outage probability if beam alignment is not frequently performed.

In this paper, we address two distinct but related topics. First, we research the design of millimeter wave wireless backhaul systems for supporting picocell data traffic.
Our focus is on the urban picocell deployment scenario where the wireless backhaul antennas are mounted on poles with link distances of 50 to 100 meters.
We propose a high gain, but computationally efficient, beam alignment technique that samples the channel subspace adaptively using subcodebook sets within a constrained time.
The proposed adaptive beam alignment algorithm utilizes a hierarchical beamforming codebook set to avoid the costly exhaustive sampling of all pairs of transmit and receive beams (which is here referred to as a non-adaptive joint alignment). 
The design of the hierarchical codebook uses a covering distance metric, which is optimized by adjusting steering squinting and utilizing efficient subarraying techniques \cite{trees2002detection}.
The proposed framework adaptively samples subspaces and searches for the beamformer and combiner pair that maximizes the receive SNR. It is shown to outperform both the non-adaptive joint alignment and the single-sided alignment (e.g., IEEE 802.11ad  \cite{IEEE:80211ad}).

Second, to motivate the practical deployment of these picocells, we also investigate wind effects on beam misalignment.
Pole sway and movement have long been studied in the civil engineering literature \cite{Piersol:2009,simiu:1996wind}, but to the best of our knowledge there has been no such work in the wireless communications research area.  We show that pole movement analysis can be used to perform backhaul failure analysis and to determine how often beam alignment must be performed. Generally, the larger the array, the more sensitive the millimeter wave link will be to wind induced misalignment.  This documents that there will likely be limitations to the achievable beamforming gain that can not be overcome by employing  larger antenna arrays.   

The paper is organized as follows. Section \ref{SysOveview} describes the system setup and the problem formulation. In Section \ref{PerfAnalysis}, the performance of the proposed beam-alignment system is analyzed in terms of the beam misalignment probability and the beam outage probability incorporating the wind-induced vibration, which is based on the established concept of beam coherence time in outdoor wireless backhaul networks.
In Section \ref{SubSampling}, a beam alignment technique is proposed that uses adaptive subspace sampling. The performance of the proposed scheme is evaluated in Section \ref{SimulationStudy}. The paper is concluded in Section \ref{Conclusions}.

Notation: a bold capital letter $\bA$ denotes a matrix, a bold lowercase letter $\ba$ denotes a vector, $\bA^T$ denotes the transpose of a matrix $\bA$, $\bA^*$ denotes the conjugate transpose of a matrix $\bA$, $\|\bA\|_F$ denotes the matrix Frobenious norm, $\|\ba\|$ denotes the vector 2-norm, and $\cC\cN(\ba,\bA)$ denotes a complex Gaussian random vector with mean $\ba$ and covariance matrix $\bA$. $\text{card}(\cA)$  denotes the cardinality of set $\cA$, $\text{rank}(\bA)$ denotes the rank of a matrix $\bA$, and $\Gamma(z)=\int^\infty_0 t^{z-1} e^{-t} dt$ denotes the gamma function.

\section{System Overview and Motivation} \label{SysOveview}
\vspace{-0.25cm}
The growth in demand for mobile broadband necessitates technical innovations in wireless network design and spectrum usage.  Generally, cellular operators have defined the usable spectrum to be roughly 300 MHz to 3 GHz \cite{Staple:2004Spectrum}.  Obviously, this small swath of  spectrum holds good propagation and implementation characteristics for many radio applications in addition to mobile broadband and it is therefore in high demand worldwide. The apparent key to using spectrum holdings as efficiently as possible is to increase frequency reuse across a geographic area (i.e., increasing the number of bits per second per Hertz per unit area) \cite{Xiaoxin:2004HierarchicalCell}.  This necessitates the use of smaller cells and dramatically increases backhaul complexity.\vspace{-0.7cm}

\begin{figure}[htbp]
    \centering
        \includegraphics[height=1.55in]{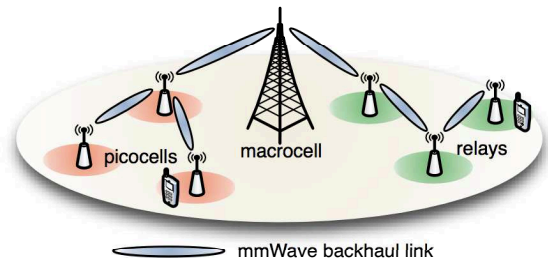}
		\vspace{-0.7cm}
    \caption{Multi-tiered cell using wireless backhaul.}
    \label{fig:MultiTieredNetwork}
\end{figure}
\vspace{-0.7cm}
Millimeter wave frequencies, roughly defined as bands between 60 GHz to 100 GHz, hold much potential for use as wireless backhaul between small cell access points and access within cells.
 A potential multi-tier cell deployment with millimeter wave wireless backhaul is shown in  Fig. \ref{fig:MultiTieredNetwork}.
 In this network, each picocell node combines its backhaul data with that received from other nodes in the network before forwarding it to the macrocell aggregation point shown. Picocell access points are expected to be separated by less than 100 meters, mitigating the deleterious effects of oxygen absorption and rain attenuation.
Coverage within the small cells (i.e., user access)  could also  be provided by millimeter wave, reducing the interference level experienced on the sub-3 GHz frequency bands used for mobile broadband.

The severe path loss of outdoor millimeter wave systems is a critical problem, especially when compared with the path losses found in other wireless systems using frequencies below 3 GHz.
In comparison with indoor millimeter wave systems, outdoor millimeter wave systems use longer links and require a much higher gain. 
\begin{figure}[t!]
    \centering
        \includegraphics[height=2.2in]{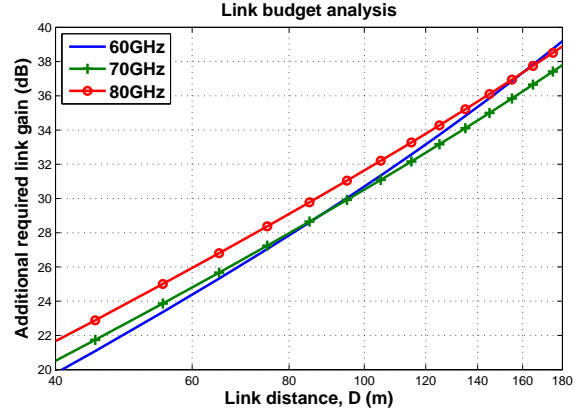}
		\vspace{-0.8cm}
    \caption{Link budget analysis for the required gain versus link distance.}
    \label{fig:LinkGain}\vspace{-0.3cm}
\end{figure}
\begin{table}[t!]
	\vspace{-0.1cm}
	\caption{ System link-budget }
	\begin{center}
	\resizebox{0.75\columnwidth}{!}{		
	\begin{tabular}{l|c}
		\hline
		Transmit Power, $P_t$ & 15 dBm \\
		Noise figure & 6 dB \\
		Thermal noise & -174 dBm/Hz \\
		Bandwidth, $B$ & 2 GHz \\
		Required SNR (QPSK with FEC) & 5 dB \\
		Pathloss model, $PL(D)$ &  $\!32.5 \!+ \!20\log_{10}(f_c) \!+ \!10\!\cdot\! a \!\cdot \log_{10}\left( D/1000\right) \!+ \!A_{i} \!\cdot \!D/1000$ dB\\
		Pathloss exponent, $a$ & 2.2 for LOS path \cite{Correia:1997kz} \\ 
		Additional Pathloss ($O_2$ and rainfall), $A_{i}$ & 20$\sim$36 dB/km \cite{Marcus:2005kv} \\
		\hline
	\end{tabular}	
	}		
	\label{table:LinkParmeters}
	\end{center}
	\vspace{-1.2cm}
\end{table}
Using a link-budget analysis, we can calculate the total required link gain for a given link distance as demonstrated in Fig. \ref{fig:LinkGain}.
For example, a 100 m link in the range of the target millimeter wave system requires an additional 32 dB or more gain for reliable communications compared to an indoor millimeter wave system.
The details of the link-budget calculation are summarized in Table \ref{table:LinkParmeters}. 
To overcome these deficiencies in path loss and obtain a large beamforming gain, large array gains are needed. An $\Mt$ transmit by $\Mr$ receive antenna system could use antennas numbering in the tens, hundreds, or potentially even thousands. 
These large arrays could provide both backhaul and access, and could be mounted
on road signs, lampposts, and other traffic control structures in urban deployments.
These access nodes would be subject to significant environmental movement (e.g., wind, moving vehicles, etc.). 
Furthermore, due to short links and narrow beam widths in millimeter wave, small changes to the propagation geometry could result in pointing errors large enough to affect link performance. 

We focus on analog transmit beamforming and receive combining for the system shown in Fig. \ref{fig:SysBD}. 
Analog beamforming using digitally controlled phase shifters is essential in millimeter wave systems to minimize the power consumption and complexity of the large number of RF chains in the array. 
\begin{figure}[htbp]
  \vspace{-0.4cm}
  \centering
  \includegraphics[width=3.2in]{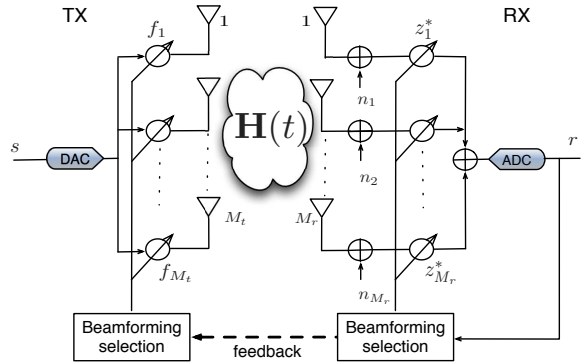}\\
  \vspace{-0.5cm}
  \caption{Block diagram showing the beamforming and combining system used on a single link in the backhaul network.}
  \label{fig:SysBD}
  \vspace{-0.5cm}
\end{figure}
In fact, duplex analog beamforming can be implemented with only a single analog-to-digital converter (ADC) and digital-to-analog converter (DAC).
The transmit data is multiplied by a transmit beamforming unit norm vector $\bff = [ f_1 \ \ f_2  \cdots  f_{\Mt}]^T \in \bbC^{\Mt}$  with $f_i$ denoting the complex weight on transmit antenna $i.$
At the receiver, the received signals on all antennas  are combined with a receive combining unit norm vector $\bz = [ z_1 \ \ z_2  \cdots  z_{\Mr}]^T \in \bbC^{\Mr}$.
The combiner output at discrete channel use $t$ given a transmit beamformer $\mathbf{f}$ and receive combiner $\mathbf{z}$ is
\begin{equation}
{r[t] = \sqrt{P} \mathbf{z}^* \left( \sum_{\tau=0}^{T-1} \bH_{\tau} s\left[t-{\tau} \right] \right)\mathbf{f}  + n[t]}
\label{eq:SysModel_Multipath}\vspace{0.2cm}
\end{equation}
where $s[t]$ is the  transmitted symbol with $E|s[t]|^2 \leq 1$,
$\bH_{{\tau}} \in \mathbb{C}^{\Mr \times \Mt} $ is the $\tau$th multiple-input multiple-output (MIMO) channel matrix among $T$ multipaths,
$n[t] \sim \mathcal{CN}(\b0, 1 )$, and $P$ represents the transmit power.

The coherence bandwidth of a millimeter wave system can be very large (e.g., on the order of 100 MHz \cite{BenDor:2011fb}), particularly in the most common line-of-sight setting.  For this reason, we make the assumption that $\mathbf{H}_1\approx\mathbf{H}_{T-1} \approx \mathbf{0}.$  This means that
\begin{equation}
\textstyle{r[t] = \sqrt{P} \mathbf{z}^*  \bH\mathbf{f} s\left[t\right]  + n[t]}
\label{eq:SysModel} \vspace{-0.1cm}
\end{equation}
using $\bH$ in place of $\bH_0$ and neglecting the path delay.
In a line-of-sight deployment, the channel $\bH$ in \eqref{eq:SysModel} can be modeled using array manifold concepts.  In this scenario, $\bH = \beta 
\ba_r(\boldsymbol{\theta}_r)\ba_t^*(\boldsymbol{\theta}_t)$ with $\beta \in \C$ representing the normalized channel gain, $\ba_r(\boldsymbol{\theta}_r)$ representing the vector in the receiver's array manifold corresponding to the angle of arrival vector $\boldsymbol{\theta}_r,$ and $\ba_t(\boldsymbol{\theta}_t)$ similarly defined.\footnote{Note that we follow the angle of departure and the angle of arrival definitions in \cite{Xu:2002vu,tse2005fundamentals} which gives a conjugate transpose.}
The receive array manifold $\cA_R$ is defined by the array geometry and the set of possible angle of arrival vectors (e.g., sector width).  The transmit  array manifold $\cA_T$ is similarly defined. 
For simplicity in notation, the subscripts for transmit and receive are removed in our general description of the array manifold. 
If the array is an $M$ element uniform linear array (ULA) with element spacing $d,$ the array manifold with possible angle of arrivals in $\cT$ is given by
\vspace{-0.25cm} \begin{equation}
\cA = \left\{\ba : \ba = \left[1
~ e^{j2\pi \frac{d}{\lambda}\sin(\theta)}~\cdots~e^{j2\pi(M-1) \frac{d}{\lambda}\sin(\theta)}
\right]^T
\phantom{a_L^T}~\text{for}~\theta\in\cT
\right\} \label{eq:1DULA}
\vspace{-0.2cm}
\end{equation}
where $d$ is the antenna element spacing, $\lambda$ is the wavelength. 
For a two-dimensional ULA with $M^2$ elements all spaced on a grid with separation $d,$ the array manifold with possible angle of arrival vectors in $\cT = \cT_1\times\cT_2$ is given by
\vspace{-0.25cm} 
\beq
\textstyle
\cA \!=\!\! \left\{\!\ba \!:\! \ba \!= \!\!\left[\left[1 \!
~ e^{j2\pi \frac{d}{\lambda}\sin(\theta)}~\!\!\cdots~\!\!e^{j2\pi(M-1) \frac{d}{\lambda}\sin(\theta)} \!
\right] \!\!\otimes \!\!\left[1 \!
~ e^{j2\pi \frac{d}{\lambda}\sin(\phi)}~\!\!\cdots~\!\!e^{j2\pi(M-1) \frac{d}{\lambda}\sin(\phi)} \!
\right]\right]^T  ~\!\!\!\!\text{for}~\!(\theta,\phi)\!\in\!\cT\!
\right\}\!\!.  \label{eq:2DURA}
\eeq
    
The employment of digitally controlled analog beamforming and combining causes a variety of practical limitations that we mathematically model \cite{springerlink:10.1007/978-94-007-0662-0_3,mailloux:2005phased}. 
An analog beamforming pattern is generated by a digitally-controlled RF phase-shifter with $q$ bits per element,  \textcolor{black}{meaning that each antenna's phase takes one value $\varphi_m$ among a size $2^{q}$ set of quantized phases given by}
\vspace{-0.15cm}\beqn
\varphi_m \in \left\{ 0, \ 2\pi\left(\frac{1}{2^{q}}\right), \ 2\pi\left(\frac{2}{2^{q}}\right)  , \ \cdots \ , 2\pi\left(\frac{2^{q}-1}{2^{q}}\right) \right\} \label{eq:finitePhaseSet}.
\eeqn 
Therefore, the beamforming vector in our analog RF beamforming system in \eqref{eq:SysModel} is defined with the quantized phases $\varphi_m$ and written as
\vspace{-0.35cm}
\beq
\bff
 = \frac{1}{\sqrt{M}} 
\left[ 1 \ \ e^{j\varphi_1} \ \ e^{j\varphi_2} \ \ \cdots \ \ e^{j\varphi_{M-1}}  \right]^T. \label{eq:beam_f}
\eeq \vspace{-0.5cm}
Then $\cF_T$ denotes the set of all beamforming vectors $\bff$ where each element phase is quantized to $q$ bits as described above and in \eqref{eq:beam_f}. Similarly, the elements of the combiner vectors $\bz$ are phase quantized and the set of possible combiners is denoted $\cZ_R$.
Furthermore, millimeter wave communication using analog beamforming and combining suffers from a subspace sampling limitation.  The receiver cannot directly observe $\bH\bf,$ rather it observes a noisy version of $\bz^*\bH\bff.$ This can be a major limitation during channel estimation and beam alignment.
In conventional MIMO beamforming systems \cite{Love:2003bc}, the beamforming codeword is selected as a function of the estimated channel to maximize some measure of system performance.
Here, however, it is not practically possible to estimate all elements of the channel matrix $\bH$. Without the full CSI of the MIMO channel and the direct estimation of the channel matrix, the problem is converted to a general problem of subspace sampling for beam alignment. The transmitter and the receiver must collaborate to determine the best beamformer-combiner pair during beam-alignment by observing subspace samples.

Note that not all beamformers in $\cF_T$ may be allowable during data transmission due to regulatory constraints on the bands employed (e.g., transmit beam width constraints during transmission).  For this reason, we will introduce two other sets $\cF\subseteq\cF_T$ and $\cZ\subseteq\cZ_R$ for the sets of possible beamformers and combiners, respectively, allowed during the data transmission. The distinction is that the larger sets $\cF_T$ and $\cZ_R$ can be used during beam alignment (which takes only a fraction of the total operational time).

In order to maximize both achievable rate and reliability, vectors $\bz$ and $\bff$ must be chosen to maximize the beamforming gain $\left|\bz^*\bH\bff\right|^2.$
Unfortunately, the transmitter and receiver can only perform subspace sampling by sending a training packet transmitted and received on a beamformer-combiner direction.
After combining (or correlating) the $\ell$-th training packet using the subspace pair $(\bz[\ell],\bff[\ell]),$ we model the receiver as having access to
\vspace{-0.3cm}\begin{equation}\label{training_packet}
y[\ell] = \sqrt{\rho} \bz^*[\ell]\bH\bff[\ell] +v[\ell] 
\vspace{-0.3cm}
\end{equation}
where $\rho$ is termed as the training signal-to-noise ratio and $v[\ell] \sim \cC\cN(0,1)$.
(Note that in \eqref{training_packet} the SNR term $\rho$ may not be the same as $P$ in \eqref{eq:SysModel}. This is because $\rho$ models the averaged SNR after the training sequence is match filtered.)

Just as channel estimation must be done reliably using limited time and power resources in lower frequency MIMO systems, beam alignment must maximize the beamforming gain using a small number of samples at a possibly low training SNR $\rho.$  We denote the total number of samples as $L$ and assume that $L = O(\Mr+\Mt).$  This  means that $L \ll \Mr\Mt$, making each sample valuable to system performance. This sampling could be done using feedback in frequency division duplexing (FDD) systems or using the link reciprocity available in time division duplexing (TDD) systems with minor modifications.

Most indoor millimeter wave alignment  schemes and radar based alignment schemes rely on a low complexity approach to selecting the beams that we refer to as \textit{hard} beam alignment.  In this technique, the selected beam pair $(\bz,\bff)$ is limited to beam pairs that have been sounded during the sampling phase.

To enforce the constraint that $\bz\in\cZ$ and $\bff\in\cF,$ we disregard samples $y[\ell]$ when $\bz[\ell]\notin\cZ$ or $\bff[\ell]\notin\cF.$  This can be succinctly written by introducing
\vspace{-0.15cm}
\begin{equation}
\widetilde{y}[\ell] = \begin{cases}
y[\ell], & \text{if $\bz[\ell]\in\cZ$ and $\bff[\ell]\in\cF$,}\\
0, & \text{otherwise.}
\end{cases}
\end{equation}
Using this function, the hard alignment algorithm returns
\vspace{-0.4cm}
\begin{equation}
(\bz_{opt},\bff_{opt}) = (\bz_{\ell_{opt}},\bff_{\ell_{opt}})~~~~~\text{with}~~~~~\ell_{opt} =  \argmax_{\ell} \left|\widetilde{y}[\ell]\right|^2. \label{B1.1}
\end{equation}

The selection of $\bz$ and $\bff$ using the small number of observed training packets can be broken into two related sub-problems:\\
\noindent\textbf{Beam Alignment Problem:}
This defines the problem of selecting the subspace pair $(\bz,\bff)$ with  $\bz\in\cZ$ and $\bff\in\cF$ to maximize $\left|\bz^*\bH\bff\right|^2$ using only the observations $y[1],\ldots,y[L].$  \\
\noindent\textbf{Subspace Sampling Problem:}
This defines the problem of selecting the subspace pair $(\bz[\ell],\bff[\ell])$ with $\bz[\ell]\in\cZ_R$ and $\bff[\ell]\in\cF_T$ for each time $\ell.$  This sampling can be done either without adaptation (i.e., $(\bz[\ell],\bff[\ell])$  is chosen independently of
$y[1],\ldots,y[\ell-1]$) or adaptively (i.e., $(\bz[\ell],\bff[\ell])$  is chosen as a function of
$y[1],\ldots,y[\ell-1]$).

\section{Performance Analysis of the Beam Alignment Criterion} \label{PerfAnalysis}
\vspace{-0.2cm}
In this section, we characterize the performance of the beam alignment in terms of the pairwise error probability. Furthermore, we model the wind-induced vibration for a small cell mounted on a lamppost, and the effect of the vibration on the beam alignment performance is measured in terms of the beam outage probability and the beam coherence time.
\vspace{-0.5cm}
\subsection{Performance Analysis of Beam Alignment} \label{Perf}
Due to the size of the arrays involved and properties of millimeter wave propagation (as discussed in Section \ref{SysOveview}), the rank of $\bH$ will be highly constrained.  Specifically, most environments will have the property that $\frac{\text{rank}(\bH)}{\min(\Mr,\Mt) }  \approx 0.$  In the majority of the cases considered for line-of-sight backhaul, it is likely that an accurate model for the channel uses $\text{rank}(\bH) = 1$\cite{HongZhang:2010wf}, especially for the scenario when all antennas are single-polarized. For this reason, we assume $\bH$ as rank one throughout our analysis. Furthermore, we model $\bH$ as being constrained so that $E\left[\|\bH\|_F^2\right] = \Mr \Mt,$ meaning that path loss is lumped into the signal transmit power or noise power term. We assume that the beam alignment is sufficiently dense that we can assume
\vspace{-0.5cm}
\begin{equation}\label{eq_assume}
\bH =  
\bh\bg^*
\end{equation}
with $\beta = 1$ in channel model $\bH$ for simplicity. 

With these assumptions, each $y[\ell]$ corresponds to a noisy observation taken using a subspace pair $(\bz[\ell],\bff[\ell]).$
For convenience, we will assume that $\bz[\ell]\in\cZ$ and $\bff[\ell]\in\cF$ for all $\ell=1,\ldots, L.$
The optimal pair of sounding vectors is denoted by $\hat{\ell}_{opt}$ and is defined as
\vspace{-0.4cm}
\begin{equation}\notag
\hat{\ell}_{opt} = \argmax_{1\leq \ell \leq L} \left|\bz^*[\ell]\bh\bg^*\bff[\ell]\right|^2.
\end{equation}
This represents the pair of vectors that would be chosen if noiseless sounding was performed.
We assume a uniform prior distribution on the optimal sounding vectors (i.e., $\hat{\ell}_{opt}$ is uniformly distributed in $\{1,\ldots,L\}$ ).
 Given this, we can now evaluate the probability of beam misalignment.
The probability of beam misalignment is expressed as
\vspace{-0.2cm}
\beq
P_{mis} &=& \frac{1}{L} \sum_{\hat{\ell}_{opt}=1}^{L} \Pr\lp \bigcup_{\hell \neq \hat{\ell}_{opt}}^{L}  \ls \la y\lS \hat{\ell}_{opt} \rS \ra^2 < \la y\lS \hell \rS \ra^2 \rs \rp. \label{Pmis}
\eeq
We can bound $P_{mis}$ by
\vspace{-0.3cm}
\beq
P_{mis} \geq 
\max_{\ell_{opt}\neq \hat{\ell}_{opt}} \Pr\lp \la y\lS \hat{\ell}_{opt} \rS \ra^2 - \la y\lS \hell\neq \hat{\ell}_{opt} \rS \ra^2 < 0  \rp,
\eeq
and
\beq 
P_{mis} \leq \frac{1}{L} \sum_{\hat{\ell}_{opt}=1}^{L} \sum_{\hell \neq \hat{\ell}_{opt} }^{L} \Pr\lp \la y\lS \hat{\ell}_{opt}\rS \ra^2 - \la y\lS \hell \rS \ra^2 < 0  \rp. \label{Pmis_UnionBound} \vspace{0.1cm}
\eeq
The above makes clear that studying the beam misalignment rate is equivalent to characterizing the
pair-wise beam misalignment probability as the bounds coincide as $\rho$ increases.

Notice that $y\lS \hat{\ell}_{opt} \rS$ and $y\lS \hell \rS$ are complex Gaussian
distributed with 
\vspace{-0.3cm}
\begin{equation}\notag
E\lS y\lS \hat{\ell}_{opt} \rS \rS = \sqrt{\rho}
\bz^*\lS \hat{\ell}_{opt} \rS \bh \bg^* \bff\lS \hat{\ell}_{opt} \rS, \ \ \
\end{equation}
\begin{equation}\notag
E\lS y\lS \hell \rS \rS = \sqrt{\rho}
\bz^*\lS \hell \rS \bh \bg^* \bff\lS \hell \rS,\vspace{-0.2cm}
\end{equation} and
$\var\lS y\lS \hat{\ell}_{opt} \rS \rS = \var\lS y\lS \hell \rS \rS = 1$.
A general expression for the probability that the difference of the two magnitudes of complex Gaussian random variables is negative
can be found in \cite{Proakis} (see Appendix B in \cite{Proakis}).
Denoting $ \gam_{\hat{\ell}_{opt}}  \triangleq \la \bz^*\lS \hat{\ell}_{opt} \rS \bh \bg^* \bff\lS \hat{\ell}_{opt} \rS \ra $ and
$\gam_{\hell} \triangleq \la \bz^*\lS \hell \rS \bh \bg^* \bff\lS \hell \rS \ra$ and
incorporating \cite{Proakis} yields the pair-wise beam misalignment probability
\vspace{-0.1cm}
\beq
\Pr\lp \la y\lS \hat{\ell}_{opt} \rS \ra^2 \!\! - \!\!  \la y\lS \hell \rS \ra^2 <0  \rp
= Q_1\lp \sqrt{\rho} \gam_{\hell} ,  \sqrt{\rho} \gam_{\hat{\ell}_{opt}} \rp 
 -\frac{1}{2} I_0\lp \rho \gam_{\hell} \gam_{\hat{\ell}_{opt}} \rp
e^{\lp -\frac{1}{2}\rho \lp \gam_{\hell}^2 + \gam_{\hat{\ell}_{opt}}^2 \rp  \rp} \label{B5}
\eeq
where $I_n\lp x \rp$ represents the modified Bessel function of the first kind,
\vspace{-0.3cm}
\beq
I_n\lp x \rp = \frac{1}{2\pi} \int_{0}^{2\pi} e^{\pm j n \theta} e^{x\cos(\theta)} d\theta , \nonumber
\eeq
and $Q_1(a, b)$ denotes the Marcum Q function,
\vspace{-0.3cm}
\beq
Q_1(a, b) = \int_{b}^{\infty} x e^{-\lp x^2 + a^2 \rp/2} I_0\lp ax \rp dx. \nonumber
\eeq
The expression in \eqref{B5} could be approximated by a closed-form expression using the results in \cite{Sofotasios:MarcumQfunc}, i.e.,
\beq
\Pr \lp \la y\lS \hat{\ell}_{opt} \rS \ra^2 \!\! - \!\!  \la y\lS \hell \rS \ra^2 < 0  \rp  
 \approx && \!\!\!\!\!\!\!\! \exp \left(-\frac{\rho}{2}\gam_{\hell}^2\right) \times \nonumber\\
 \sum^{k}_{l=0} 
\frac{\Gamma(k\!+\!l) \!\ k^{1-2l}  \rho^l \gam_{\hell}^{2l}}
{\Gamma^2(l+1) \Gamma(k-l+1) 2^l} && \!\!\!\!\!\!\!\!\!
\!\!\left[  \Gamma\!\left(\!l\!+\!1,\!{ \frac{\rho }{2}\gam_{\hat{\ell}_{opt}}^2 }\!\right)\! 
\!-\!\frac{\rho^l \gam_{\hat{\ell}_{opt}}^{2l}}{2^{l+1}\exp\left(\frac{\rho}{2} \gam_{\hat{\ell}_{opt}}^2 \right)}
\right]\!\!, \ \ \ \ \ \ \ \ \ \ \ \ \ \ \label{Prob_err} 
\eeq 
where $\Gamma(x,y)$ denotes the upper incomplete gamma function,
\beq
\Gamma(a,x) = \frac{1}{\Gamma(a)}\int^\infty_x e^{-t} t^{a-1} dt.
\eeq

Directly analyzing the expression in \eqref{Prob_err} does not provide much intuition. For this reason, we incorporate a tractable expression of \eqref{B5} by introducing the notation \cite{Pawula}
\beq
P(U, V) &=& Q_1\lp \sqrt{U-W}, \sqrt{U+W} \rp -\frac{1}{2} e^{-U} I_0(V) \label{B5.0}
\eeq
where we have
\vspace{-0.3cm}
\beq
U=\frac{1}{2}\rho \lp \gam_{\hat{\ell}_{opt}}^2 + \gam_{\hell}^2 \rp, \ 
V = \rho \gam_{\hat{\ell}_{opt}} \gam_{\hell} , \ 
W = \frac{1}{2}\rho \lp \gam_{\hat{\ell}_{opt}}^2  - \gam_{\hell}^2 \rp.  \label{B5.1}
\eeq
As $U,V$ tend to infinity while keeping $U\geq V$, $P(U, V)$ in \eqref{B5.0} converges to \cite{Pawula}
\beq
P(U,V) \stackrel{U,V\rightarrow \infty}{=} \sqrt{\frac{U+V}{8V}} \textrm{erfc}\lp \sqrt{U-V} \rp  \label{B5.1.1}
\eeq
with $\textrm{erfc}(x)$ denoting the complementary error function.\footnote{The complementary error function is defined as $\textrm{erfc}(x)=\frac{2}{\sqrt{\pi}} \int^\infty_x e^{-k^2}dk$.}
Notice that $U$ and $V$ in \eqref{B5.1} always satisfy $U\geq V$. 
\textcolor{black}{Keeping the exponential dependency, \eqref{B5.1.1}} can be simplified to
\beq
\color{black} 
P(U, V) \stackrel{U,V\rightarrow \infty}{=} \sqrt{\frac{U+V}{8V}} e^{-(U-V)} \label{B5.1.2}
\eeq
where in \eqref{B5.1.2} we use the fact that $\textrm{erfc}\lp x \rp \approx e^{-x^2}$ as $x$ tends infinity.

The asymptotic expression in \eqref{B5.1.2} readily allows us to obtain the expression for the pair-wise beam misalignment rate
\vspace{-0.4cm}
\beq
\color{black}
\textstyle\Pr\lp \la y\lS \hat{\ell}_{opt} \rS \ra^2 - \la y\lS \hell \rS \ra^2 <0  \rp
\stackrel{\rho, \gam_{\hat{\ell}_{opt}}, \gam_{\hell}\rightarrow \infty }{=} \sqrt{ \frac{\lp \gam_{\hat{\ell}_{opt}} + \gam_{\hell} \rp^2 }{8(\gam_{\hat{\ell}_{opt}}\gam_{\hell})} }
e^{-\frac{\rho}{2} \lp \gam_{\hat{\ell}_{opt}}  - \gam_{\hell} \rp^2}. \label{B5.2}\vspace{0.2cm}
\eeq
Notice that since $\gam_{\hat{\ell}_{opt}}>\gam_{\hell}$, the beam alignment ensures
a pairwise \emph{exponential decay} of the beam misalignment rate \textcolor{black}{as $\rho$, $\gam_{\hat{\ell}_{opt}}$, and $\gam_{\hell}$ increase.}

We validate the analysis presented in this section in Fig. \ref{fig:PerfAnalysis}. The system is assumed to have $M=32$ transmit and receive antennas at each side. A size $64$ beamforming and combining codebook is used for the numerical simulation. The channel vector $\bh$ and $\bg$ are modeled as arising from a ULA in \eqref{eq:1DULA}. 
To demonstrate the accuracy of \textcolor{black}{the bounds} and the asymptotic expressions derived, the plots of beam misalignment using \eqref{B5.1.1} and \eqref{B5.2} are compared with the plots of beam misalignment in \eqref{Pmis} and upper bound in \eqref{Pmis_UnionBound}.
\textcolor{black}{As seen from the figure, the pair-wise misalignment probability in \eqref{Pmis_UnionBound} coincides with $P_{mis}$ in \eqref{Pmis} as SNR increases. Furthermore, it is clear that the asymptotic expression in \eqref{B5.1.1} indeed tightly models \eqref{Pmis_UnionBound}. Notice that we also plot \eqref{B5.2} to demonstrate that \eqref{B5.2} closely models the slope behavior of the beam misalignment rates in \eqref{Pmis} and \eqref{Pmis_UnionBound}.}
\begin{figure}[htbp]
  \vspace{-0.3cm}
  \centering
  \includegraphics[height=2.45in]{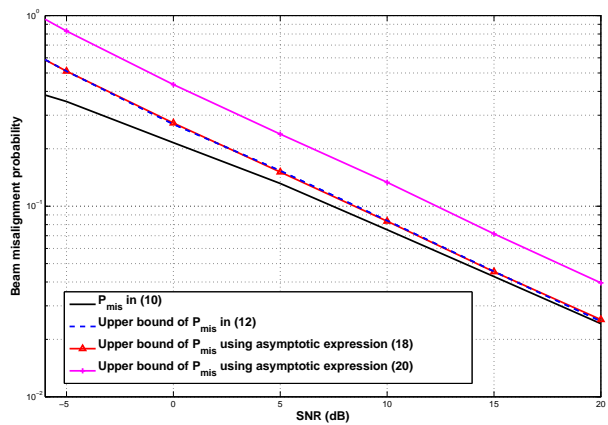}
  \vspace{-1.0cm}
  \caption{Comparison beam misalignment probabilities for $M=32$ with $\textrm{card}(\cF)=\textrm{card}(\cZ)=64$.}
  \label{fig:PerfAnalysis}
  \vspace{-0.7cm}
\end{figure}

\subsection{Wind Induced Impairments in Beam Alignment} \label{WindSwayEffect}
In practical scenarios, small cells deployed in urban outdoor environments are regularly affected by wind. In millimeter wave beamforming systems, the wind-induced movement is on the order of hundreds of wavelengths and they use a very narrow beam pattern. We consider the practical impairments of a lamppost deployment scenario by modeling the wind-induced vibration and incorporating the wind-sway analysis methodology from the civil engineering literature into our beamforming system and system design. The details are summarized in Appendix \ref{WindModel}.

\begin{figure}[htbp]
  \vspace{-0.4cm}
  \centering
  \includegraphics[height=3.5in]{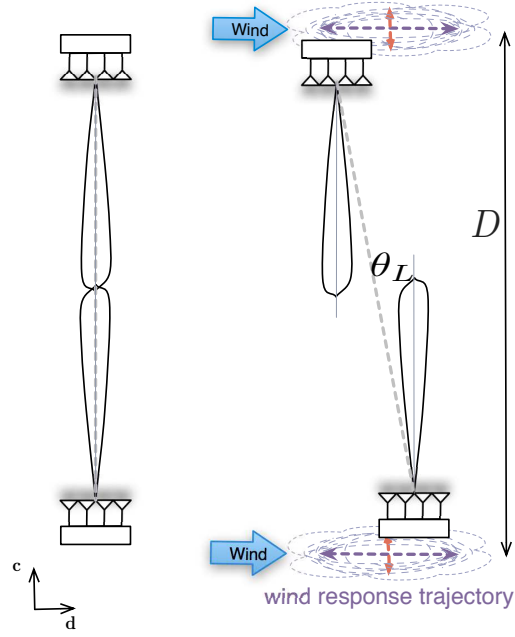}
  \vspace{-0.5cm}
  \caption{Example pole movement showing beam deflection due to wind.}
  \label{fig:BeamOutage}
  \vspace{-0.7cm}
\end{figure}
Following the development of the wind-sway model in Appendix \ref{WindModel}, a trajectory for the motion of the top of the antenna mounting pole can be computed. Assuming that the wind turbulence components are independent for the poles at the two ends of the link we may generate trajectories and misalignments as illustrated in Fig. \ref{fig:BeamOutage}. If $\Delta L_d(t)$ and $\Delta L_c(t)$ denote the relative displacements of the pole-tops at the two ends of the link and assuming a link distance of $D$ and the worst case scenario where the mean wind direction is perpendicular to the beam direction then the sway angle is given by
\beq
\theta_L(t) =\tan^{-1} \left( \frac{\Delta L_d(t)}{D + \Delta L_c(t)}\right).
\label{eq:theta_D_L}
\eeq
Power fluctuation due to $\Delta L_c(t)$ variation is negligible due to large link distance $D$ so the outage probability depends only on the beam angular deflection.
Define beam outage as the event where the beam deflection angle $\theta_L(t)$ is larger than some maximum allowable deflection angle $\theta_{L,max}$. Then the beam outage probability is given by
\vspace{-0.5cm}
\begin{equation}
P_{out} = Prob \{ | \theta_L(t)| \ > \ \theta_{L,max}| \}.\vspace{-0.2cm}\label{eq:Pout}
\end{equation}
If $T_{out}$ is a random variable representing the time to first outage, i.e.,
\vspace{-0.4cm}
\begin{equation}
 T_{out} = \inf \{T: |\theta_L(t)|\leq\theta_{L,max} \ \ \textrm{for} \ \ 0 \leq t < T \ \ \textrm{and} \ \ |\theta_L(T)| > \theta_{L,max}\} 
\end{equation}
then the coherence time is defined to be $T_c = E[T_{out}]$.

Numerical results were obtained for the following parameters:
\vspace{-0.2cm}
\begin{itemize}
\item Link distance $D = 50$ m.
\vspace{-0.1cm}
\item Mean wind speed $\overline{u} =  13$ $\mbox{m/s}$.
\vspace{-0.1cm}
\item Air density $\rho_a = 1.22$ $\mbox{kg/m}^3$.
\vspace{-0.1cm}
\item Coefficient of drag $C_D = 0.5$. 
\vspace{-0.1cm}
\item Pole response parameters $f_n = 1 \ Hz$, $\zeta = 0.002$ \cite{simiu:1996wind}, mass of the pole and antenna \vspace{-0.1cm} 

mounting is considered to $m = 5$ kg, and effective area $A_e = 0.09$ $\mbox{m}^2$.
\end{itemize} \vspace{-0.1cm}

The maximum deflection angle $\theta_{L,max}$ is defined as a small fraction of the beam width $\theta_{BW}$  given by 
\vspace{-0.7cm}
\beq
\theta_{L,max} = \alpha \theta_{BW}
\eeq
where $\alpha$ is the fraction ratio calculated for a certain beamforming gain loss. 
For a uniform linear array with half-wavelength spacing $\theta_{BW} = 2\sin^{-1}\left({0.891}/{M} \right)$ in \cite{trees2002detection}.
For a 3 dB loss in two-sided ULA beamforming gain, $\alpha = 0.3578$ is obtained from the relationship between the beamforming pattern and the beam width $\theta_{BW}$ using standard array parameter values.

In Fig. \ref{fig:subfigC1}, the beam outage probability in \eqref{eq:Pout} for various array sizes is plotted as a function of the mean wind speed
$\overline{u}$. For example, an $M=32$ system in strong wind turbulence with $\overline{u}=20 \ \mbox{m/s}$ is in outage approximately 25 percent of the time. Note that a large-sized array is much more sensitive to beam deflection due to wind.  This means that the achievable beamforming gain could be limited no matter the array size because of this wind-induced beam misalignment. To overcome this misalignment, beam realignment will have to be done frequently.  
Similarly, the coherence time is given in Fig. \ref{fig:subfigC2}.
We are interested in the order of the beam coherence time. 
Notice that the expected beam coherence time of the $M=64$ system with $\overline{u}=20 \ \textrm{m/s}$ is on the order of 100s of milliseconds. From this modeling, a system that needs to track the beam would require an alignment search time somewhat smaller than the order of milliseconds to avoid beam outage. Many practical settings must deal with moving vehicles on streets and Doppler frequency shifts caused by scatters, and the systems found in these settings require more frequent alignment to satisfy a smaller beam coherence time.
\begin{figure}[htbp]
\subfigure[]
{
	\includegraphics[height=2.5in]{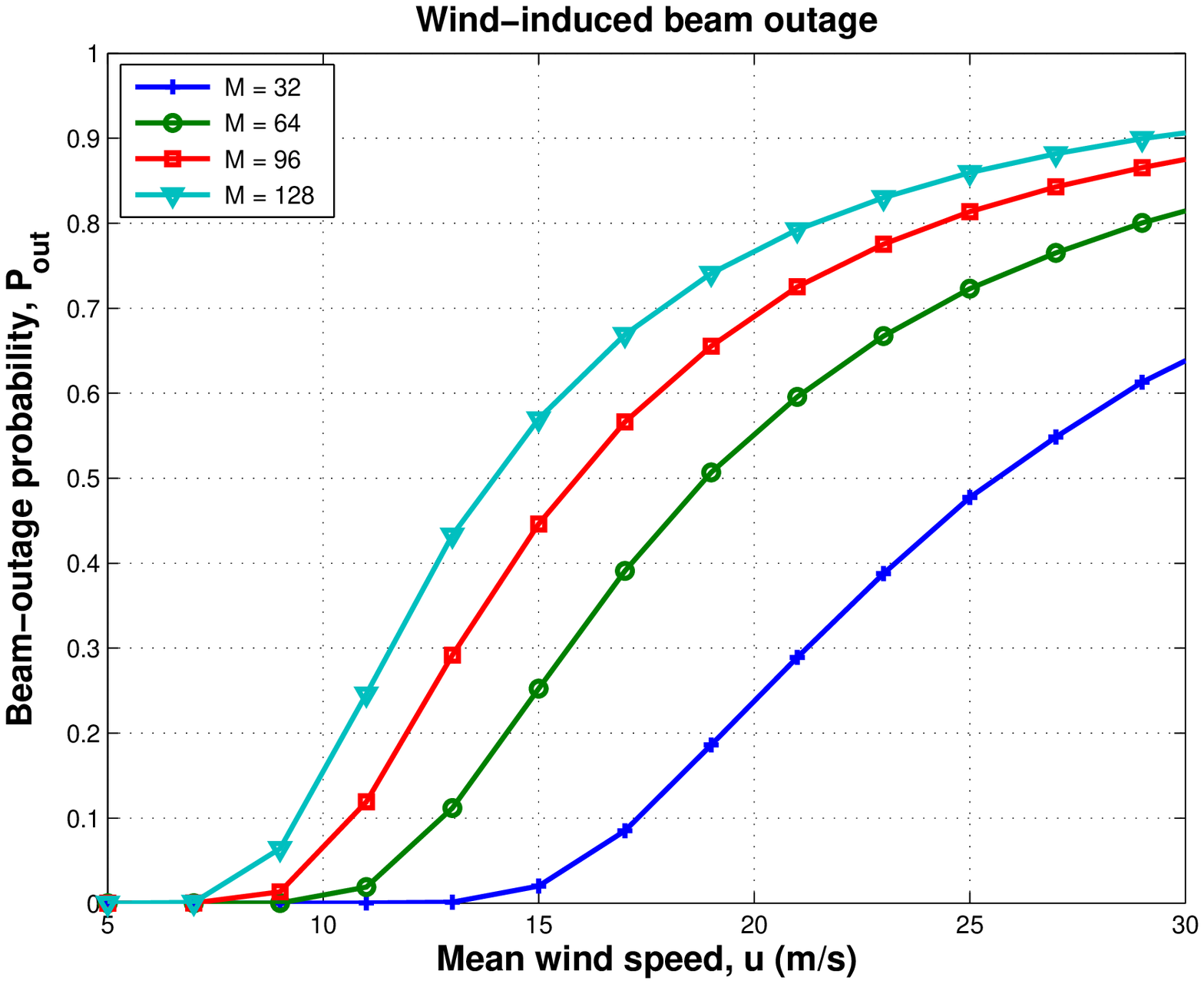}
    \label{fig:subfigC1}
}\vspace{-0.3cm}
\subfigure[]
{
	\includegraphics[height=2.5in]{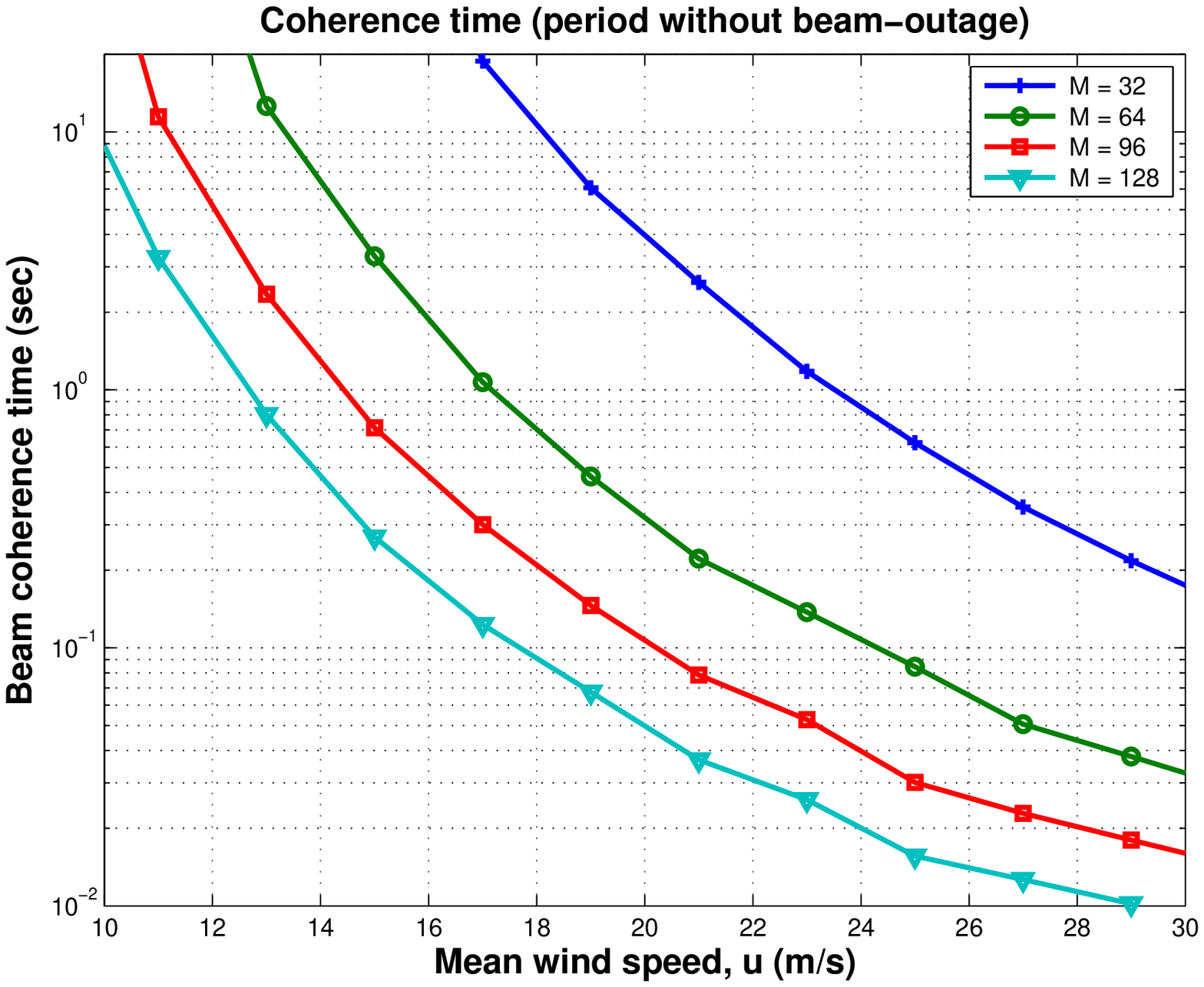}
    \label{fig:subfigC2}
}
\caption{\subref{fig:subfigC1} Beam outage probability.
		 \subref{fig:subfigC2} Beamforming coherence time on wind speed.}
\label{fig:WindSimulation}
\end{figure}

\vspace{-0.5cm}
\section{Subspace Sampling for Beam Alignment} \label{SubSampling}
\vspace{-0.2cm}
Though the actual beam selection algorithms in the previous section are important, they are limited to the observed data.  For this reason, it is critical that the sampled subspaces are chosen judiciously.  We overview both non-adaptive  and adaptive subspace sampling.
\vspace{-0.7cm}
\subsection{Non-Adaptive Subspace Sampling}\label{randsampling}
The most time intensive, but most obvious, method of sampling is to simply sound the channel with all possible pairs of beamforming and combining vectors.  In this method, the total sounding time required is $L = \text{card}(\cZ)\text{card}(\cF).$

If we form an $\Mr\times \text{card}(\cZ)$ matrix $\bZ_{all}$ using all vectors in $\cZ$ and an $\Mt\times \text{card}(\cF)$ matrix $\bF_{all}$ using all vectors in $\cF.$
We can collect all of our samples and write the sampled signal using $\by = \left[y[1]\cdots y[L]\right]^T$ and $\bv = \left[v[1]\cdots v[L]\right]^T$ as
\begin{equation}
\by = \sqrt{\rho}\text{vec}\left(\bZ_{all}^*\bH\bF_{all}\right)+\bv \label{eq:NonAdapSubSampling}
\end{equation}
where $\text{vec}$ stacks the columns of the matrix into a column vector.
The received vector can then be easily used for alignment. The selected beam pair then corresponds to the index that achieves the sup norm $\|\by\|_\infty$.

\subsection{Adaptive Sampling}\vspace{-0.3cm}
Most practical scenarios allow the transmitter access to some information (e.g., through feedback) about ${y[1],\ldots,\!y[\ell\!-\!1]}$ prior to sending training packet $\ell.$  This can allow adaptive sampling and the potential to dramatically increase beamforming gain by overcoming noise during training.

\noindent\textbf{Ping-Pong Adaptive Sampling}\\
Consider $y[\ell]$ in (\ref{training_packet}).  Clearly if $\bff[\ell]$ is close in subspace distance to the dominant right singular vector of $\bH,$ we can obtain a high SNR estimate of the left singular vector of $\bH.$  Similarly, if $\bz[\ell]$ is close in subspace distance to the dominant left singular vector of $\bH,$  we can obtain a high SNR estimate of the right singular vector of $\bH.$

These observations motivate \textit{ping-pong sampling}.  Let $\bff_{opt,\ell}$ and $\bz_{opt,\ell}$ denote the estimated beam directions (using some beam alignment algorithm) using samples $y[1],\ldots,y[\ell].$ 
In $K$-round ping-pong sampling with $L_K = \frac{L}{2K}$ assumed to be a positive integer, 
$\bff[\ell] = \bff_{opt,2 L_K\lceil{\ell/2L_K\rceil}-L_K}$ when  ${(\ell-1) \mod 2 L_K} \geq L_K$
and
$\bz[\ell] = \bz_{opt,2 L_K\lfloor{\ell/2L_K\rfloor}}$ when ${(\ell-1) \mod 2 L_K} < L_K.$ 
The basic idea is to allow the transmitter to probe the channel's subspace structure with assistance from the receiver for the first half of each round and the receiver to probe the channel's subspace structure with assistance from the transmitter during the last half of each round.
Note that the initial receive beam $\bz_{opt,0}$ at the first stage of the ping-pong sampling can be defined as any initial beam. The details of the ping-ping sampling strategy are shown in Fig. \ref{fig:ping_pong}. Each bin represents the transmit and receive beams respectively used for channel sounding. The transmit beamformer $\bff[\ell]$ and receiver combiner $\bz[\ell]$ pair is selected and the output sample $y[\ell]$ is observed.
After each $L_K$ ping-pong sampling, $\bff_{opt,\ell}$ or $\bz_{opt,\ell} (\ell = L_K, 2L_K, \cdots, 2KL_K)$ are estimated, then the final pair of beam $\bff_{opt,L}$ and $\bz_{opt,L}$ is estimated using the $L$ observations.
\begin{figure}[htbp]
	\vspace{-0.6cm}
	\centering
		\includegraphics[height=1.5in]{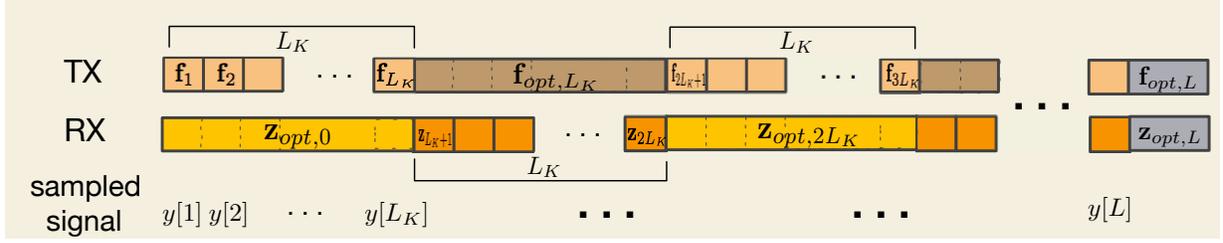}
		\vspace{-1.2cm}
	\caption{Flowchart of ping-pong adaptive sampling.}
	\label{fig:ping_pong} 
	\vspace{-0.5cm}
\end{figure}

\noindent\textbf{Adaptive Subspace Sampling using Hierarchical Subcodebooks}\\
\indent In millimeter wave systems, the potentially large number of antennas and substantial beamforming gain requirement will necessitate the codebook sizes of $\cZ$ and $\cF$ to be \textit{very} large.  For example, if $M = 100,$ which is not unreasonable, and a simple codebook is constructed by fixing the phase on one antenna and phase shifting every other antenna by one of four unique phases ($q=2$), the codebook $\cF$ might be of size $4^{99}\approx 4.0173\cdot 10^{59}.$  Therefore, it would be completely impractical to search over all possible beamformer and combiner pairs. To alleviate these concerns, we consider an adaptive subspace sounding method. For the sake of simplicity, the sounding method will be described for the transmit-side beamformer using $\cF,$ but all of the described techniques equally apply to the receive-side combiner.

Assume the $K$-round ping-pong sounding approach discussed above.  We assume that $\cF$ is designed to uniformly cover (or quantize) the array manifold denoted by $\cA$.
We  construct a series of increasing resolution codebooks
$\cF_1,$ $\cF_2,\ldots,$ and $\cF_K$ satisfying $N_1= \text{card}(\cF_1) < N_2=\text{card}(\cF_2)<\cdots<N_K=\text{card}(\cF_K)$ with $\cF_K = \cF$ and $\cF_k \subset \cF_T$ for all $k$.   
The codebook sizes can be flexibly defined with the only requirement being that $N_k \leq (L_K)^k$ for $k=1,\ldots,K.$
Generally, the first round subcodebook $\cF_1$ is designed with $N_1 = L_K$, and the last round subcodebook $\cF_K$ is designed with $N_K \geq 2 M.$\footnote{The codebook size in millimeter wave beamforming with a large-sized array is generally more than twice as large as the number of antennas in order to keep gain fluctuation within $1\textrm{dB}$ \cite{Wang:2009jr}.}

The task then is to create subcodebooks of $\cF_T$ to quantize $\cA.$  
We can do this by minimizing the covering distance of the code over the space $\cA$. 
The covering distance is given by 
\begin{equation}
\delta\left(\widetilde{\cF}\right) = \sqrt{1- \frac{\chi(\widetilde{\cF})}{M}}
\end{equation} 
where $\chi(\widetilde{\cF})$ is the minimum absolute squared inner product of the subcodebook beams and the array manifold defined by
\begin{equation}
\chi\left(\widetilde{\cF}\right) = \min_{\ba\in\cA} \max_{\bff\in\widetilde{\cF}} |\bff^*\ba|^2. \label{eq:CodeMetric}
\end{equation}   
Equation \eqref{eq:CodeMetric} tells us the smallest beamforming gain factor possible given the codebook $\widetilde{\cF}$ and perfect selection.
We therefore pick each codebook $\cF_k, \ k=1,\ldots,K-1,$ according to
\begin{equation}
\cF_k = \argmin_{\widetilde{\cF}\subseteq\cF_T:\text{card}(\widetilde{\cF}) = N_k} \delta\left(\widetilde{\cF}\right). \label{eq:CodebookDesign}\vspace{0.2cm}
\end{equation}
Choosing the codebook in this way can be done offline. Each subcodebook will thus maximize the minimum beamforming gain possible, and the details of the subcodebook design are discussed shortly.

Given the subcodebooks, we must now determine how to traverse the subcodebooks.
After reception of round $k,$ the optimal beam $\bff_{opt,[k]}$ in $\cF_k$ is chosen according to the beam alignment algorithms using multiples of $L_K$ observations. The optimal beamformer and combiner at round $k$ are denoted by $\bff_{opt,[k]} = \bff_{opt,(2(k-1)+1)L_K}$ and $\bz_{opt,[k]} = \bz_{opt,2(k-1)L_K}$, respectively. 
To utilize the hierarchical structure between the subcodebooks, the $L_K$ beams in $\cF_{k+1}$ closest to $\bff_{opt,[k]},$  denoted by $\cF_{(k+1)\mid \bff_{opt,[k]}}$ are sounded for round $k+1$. This means
\begin{equation}
\cF_{(k+1)\mid \bff_{opt,[k]}} = \argmax_{\{\bff_{i_1},\ldots,\bff_{i_{L_K}}\}\subseteq \cF_{k+1} : i_1<\cdots<i_{L_K}}
\min_{i\in\{i_1,\ldots,i_{L_K}\}} \left|\bff_{opt,[k]}^* \bff_i\right|^2. 
\label{eq:SubcodebookSelection}\vspace{0.1cm}
\end{equation}
This method of sounding has a graphically appealing interpretation when $\cF_k$ consists of vectors formed by uniformly quantizing the possible angles of departure. The next level sounded beams in $\cF_{(k+1)\mid \bff_{opt,[k]}}$ are uniformly spanned within the sector covered by the optimal beamformer at the previous level, $\bff_{opt,[k]}$. 
Fig. \ref{fig:tree_codebook} demonstrates the structure and relationship between the subcodebooks. Each subcodebook consists of $N_k$ codewords, and the $L_K$ expansion from the selected optimal codeword in every level is shown.
In the same manner, the optimal combiner at round $k$, $\bz_{opt,[k]}$, is selected by sounding through the subcodebooks of combiners $\cZ_1,\cZ_2,\cdots,\cZ_K$.
\begin{figure}[htbp]
	\vspace{-0.5cm}
	\centering
		\includegraphics[height=2in]{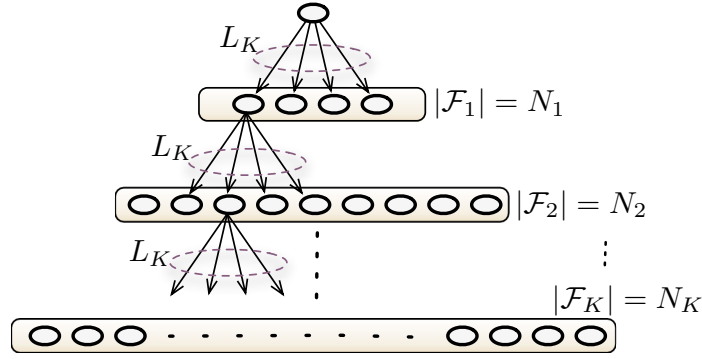}
		\vspace{-0.5cm}
	\caption{Hierarchical structure of subcodebooks $\cF_k$ and expansions between subcodebooks.}
	\label{fig:tree_codebook} 
	\vspace{-0.5cm}
\end{figure}

After adaptive $K$-round ping-pong sampling and sounding through traversing the beamformer and combiner of subcodebooks in \eqref{eq:SubcodebookSelection}, the best beam pair $\bff_{opt,L} \in \cF_K$ and $\bz_{opt,L} \in \cZ_K$ is obtained.

\noindent\textbf{Subcodebook Design}\\
\indent For efficient adaptive alignment and sampling, optimized subcodebooks for $\cF$ are required as described in \eqref{eq:CodebookDesign}. Each subcodebook provides increased beamforming gain and improved subspace sampling. The size and design of the subcodebook will influence the beamforming gain. For convenience, we restrict the discussion to a one-dimensional ULA.  However, these techniques are obviously more generally extendable to a two-dimensional ULA using the Kronecker product in \eqref{eq:2DURA}.  

First note that the covering distance inner product can be bounded.

\begin{lemma}
The covering distance inner product of the subcodebook $\cF_k = \{\bff_1,\ldots,\bff_{N_k}\}$  for an $M$ antenna one-dimensional ULA is bounded by 
\begin{equation}
\chi\left( \cF_k \right) \leq \min\left(\frac{ 2\pi N_k}{\mu(\cP)},M \right) \label{eq:Lemma1}
\end{equation}
where $\mu(\cP) = \int_{-\pi}^\pi \textbf{1} (\psi\in\cP) d\psi$ with $\textbf{1} (\cdot) $ denoting the indicator function, 
$\cP = \{ \psi:\psi = 2\pi \frac{d}{\lambda}\sin(\theta), \  \theta \in \cT \}$, and $\cT$ is the set of possible angles of departure of the ULA.
\end{lemma}
\begin{proof}
For the ULA, the beam pattern specifies the inner-product with vectors on the array manifold.  The beam pattern for a vector $\bff$ is given by $G_{\bff}(\psi) = \left|\sum_{m=0}^{M-1} f_m e^{jm\psi}\right|^2.$
By Parseval's theorem, $\frac{1}{2\pi}\int_{-\pi}^\pi G_{\bff}(\psi) d\psi = \|\bff\|^2.$ 

Because the area under the beam pattern is bounded, we can use a sectored approach to understand the covering distance.  
Let the sector region $\vartheta_i$ define the set of array manifold vectors ``closest" to $\bff_i.$  Mathematically, this means 
$\vartheta_i = \left\{\psi \in \cP : |\ba^*(\psi)\bff_i| > |\ba^*(\psi)\bff_j| \ \ \textrm{for} \ \  i \neq j\right\}.$ 
Clearly, the area under the sector of angles $\vartheta_i$ is bounded by the entire beam pattern area, which is written as
$
\int_{\vartheta_i} |\bff_i^* \ba(\psi)|^2 d\psi \leq \int_{-\pi}^{\pi} |\bff_i^* \ba(\psi)|^2 d\psi  = 2\pi.
$
Using this, the bounded beam gain on $i$th sector is given by
\vspace{-0.4cm}
\beq
\min_{\psi \in \vartheta_i}|\bff^*_i \ba(\psi)|^2 &\leq& \frac{\int_{\vartheta_i} |\bff_i^* \ba(\psi)|^2 d\psi}{\mu(\vartheta_i)} \nonumber\\
&\leq& \frac{2\pi}{\mu(\vartheta_i)} \nonumber
\eeq
where $\mu(\vartheta_i)$ is the length of angle interval $i$. The absolute inner product term $\chi(\cF_k)$ is written by the bounded area, which is given by
\beq
\chi\left( \cF_k = \left\{\bff_1, \cdots , \bff_{N_k}\right\} \right) 
&=& \max_{i\in\{1,\ldots,{N_K}\}} \min_{\psi \in \vartheta_i} \left| \bff_i^* \ba(\psi)\right|^2  \nonumber\\
&\leq& \max_{i\in\{1,\ldots,{N_K}\}} \frac{2 \pi}{\mu(\vartheta_i)} \nonumber\\
&=&  \frac{2 \pi}{\min_{i\in\{1,\ldots,{N_K}\}}\mu(\vartheta_i)} 
\label{eq:lemmaPF2}. 
\eeq
The expression in \eqref{eq:lemmaPF2} is minimized when $\mu(\vartheta_i), i=1,\ldots,{N_K}$ is equally divided as $\frac{\mu(\cP)}{N_k}$. Thus, $\chi(\cF_k)$ is bounded using the covering distance for the case of equally sized sector regions.
Additionally, the beamforming gain is bounded by $| \bff_i^* \ba(\psi)|^2 \leq M$.  Therefore, the metric $\chi(\cF)$ in \eqref{eq:lemmaPF2} satisfies 
$ \chi\left( \cF_k \right) \leq \min\left(\frac{ 2\pi N_k}{\mu(\cP)},M \right)$.
\end{proof}

The upper bound in \eqref{eq:Lemma1} corresponds to  beamforming vectors with beam patterns that do not overlap.  Though beam patterns of this form are usually unrealizable,  $\cF_1, \cF_2, \cdots, \cF_{(K-1)}$  will each be designed in an attempt to have collective beam patterns approximating those shown in  Fig. \ref{fig:QuantizedArrayManifold} assuming $\cT$ defines a sector of angle directions.
Therefore, beamforming vector $i$ in subcodebook $\cF_k$ will have a beam direction corresponding to 
\vspace{-0.3cm}
\begin{equation}
\psi_i = \psi_{LB} + \left(\frac{\psi_{UB}-\psi_{LB}}{2 N_K} \right)+ \frac{i (\psi_{UB} - \psi_{LB})}{N_K}, \ \ i = 0, 1, \ldots, N_K-1 .\label{eq:psi_i} 
\end{equation}
Physically, this corresponds to an angle-of-arrival (or departure) $\theta_i$  calculated as $\theta_i = \sin^{-1}\left( \frac{\psi_i}{2\pi d/\lambda}\right)$.
\begin{figure}[htbp]
	\vspace{-1.0cm}
	\centering
		\includegraphics[height=1.3in]{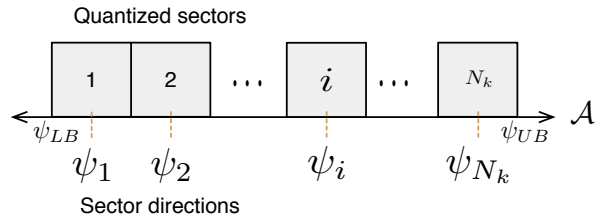}
		\vspace{-0.5cm}
	\caption{Sectorized array manifold for subcodebook design at stage $k$.}
	\label{fig:QuantizedArrayManifold} 
	\vspace{-0.5cm}
\end{figure}

After dividing sectors in $\cA$, a beam ``flattened'' to have an omni-directional pattern within each sector is desired.  Flattened omni-directional beams can be generated using techniques such as the subarray method in \cite{Mailloux:2007wm}, which also employs beam-spoiling techniques similar to \cite{Kinsey:1997eq,mailloux:2005phased}. This technique divides  the $M$ element array into $M_{sub}$ subarrays  each of size $M/M_{sub}$. 
The subbeams are each defocused with a small offset angle, then they are summed for broadening. For a one-dimensional ULA, the broadened beamformer $\bff_{br}$ is given by
\vspace{-0.3cm}
\beq
\bff_{br} = \left[ {\bff}_{comp,1}^T \ {\bff}_{comp,2}^T \ \cdots \ {\bff}_{comp,\left(M/M_{sub}\right)}^T  \right]^T \vspace{-1.5cm}
\eeq
where $\bff_{br} \in \C^{M \times 1}$ and ${\bff}_{comp,j} \in \C^{M_{sub} \times 1}$ is a subarray for a component subbeam. The component subbeams are each pointed in slightly different directions to broaden the beam before application of a defocusing angle $\theta_{sp}.$ Note that the broadened beamformer $\bff_{br}$ is a unit norm vector and each element is controlled by analog RF beamforming in \eqref{eq:beam_f}. Each broadened beam constructs subcodebook, $\bff_{br} \in \cF_k$.

An example beam pattern for a single sector assuming $M = 32$  elements is  shown in Fig. \ref{fig:subfigA1}. The array is grouped into eight subarrays (i.e., $M_{sub} = 8$) and four component subbeams ${\bff}_{comp,j}$ using a 5 bit quantized phase-shifter for each element. Then, the superposition of component subbeams is defocused and pointed in a beam direction $\theta_i = 15^\circ$.

With this approach,  we optimize $\chi(\cF_k)$ over $M_{sub}$ and $ \theta_{sp}$ to design subcodebook $\cF_k$ in \eqref{eq:CodebookDesign}. An example subcodebook for the $\Mt=32$ case with $N_k=16$ is shown in Fig. \ref{fig:subfigA2}. 
The array manifold $\cA$ is uniformly quantized to $N_k=16$ directions, and the subarray parameters are optimized to $M_{sub}^* = 2$ and $\theta_{sp}^* = 1.72^\circ$.
\begin{figure}[htbp]
\centering
\subfigure[ ]{
	\includegraphics[height=1.70in]{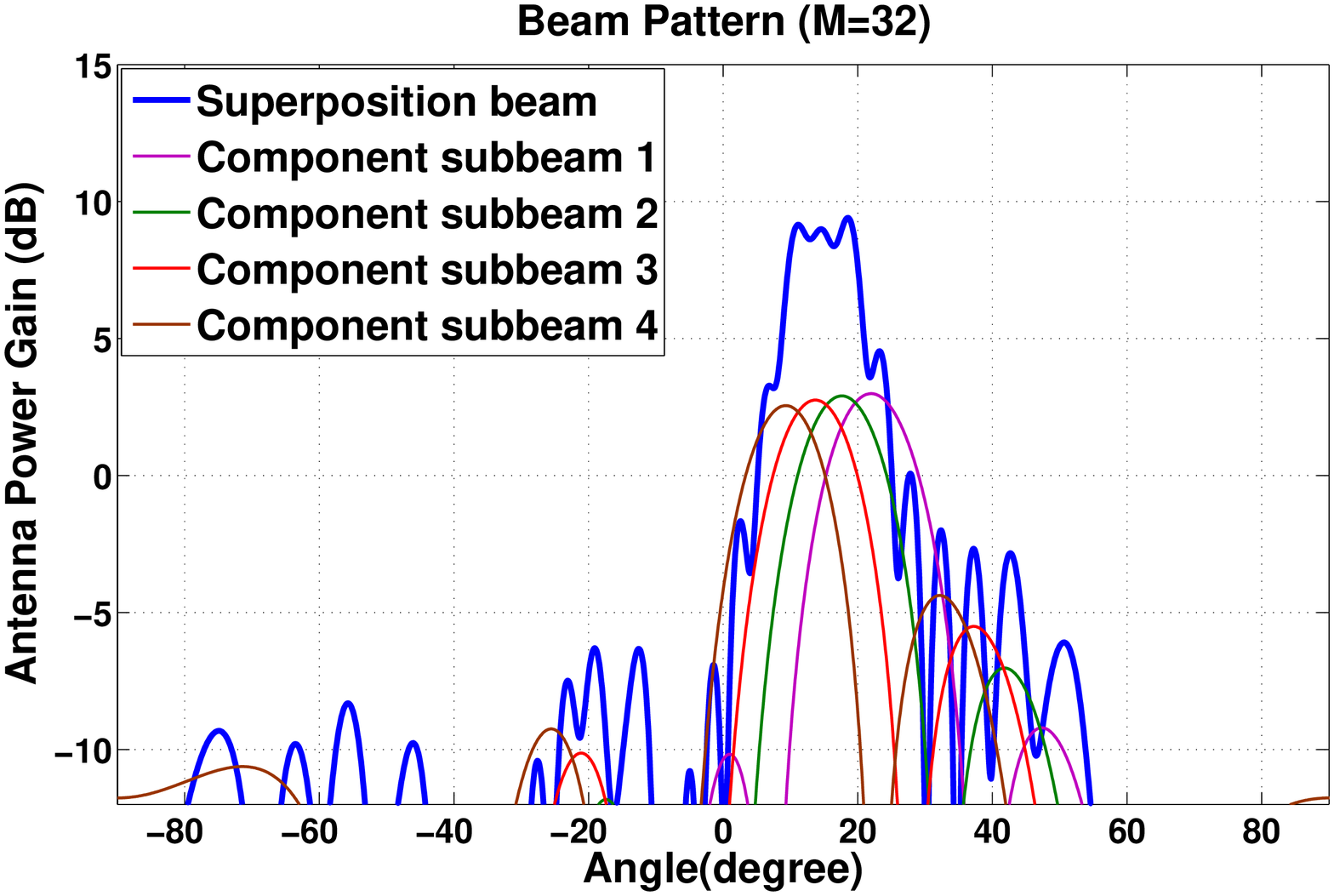}
    \label{fig:subfigA1}
}
\subfigure[ ]{
	\includegraphics[height=1.8in]{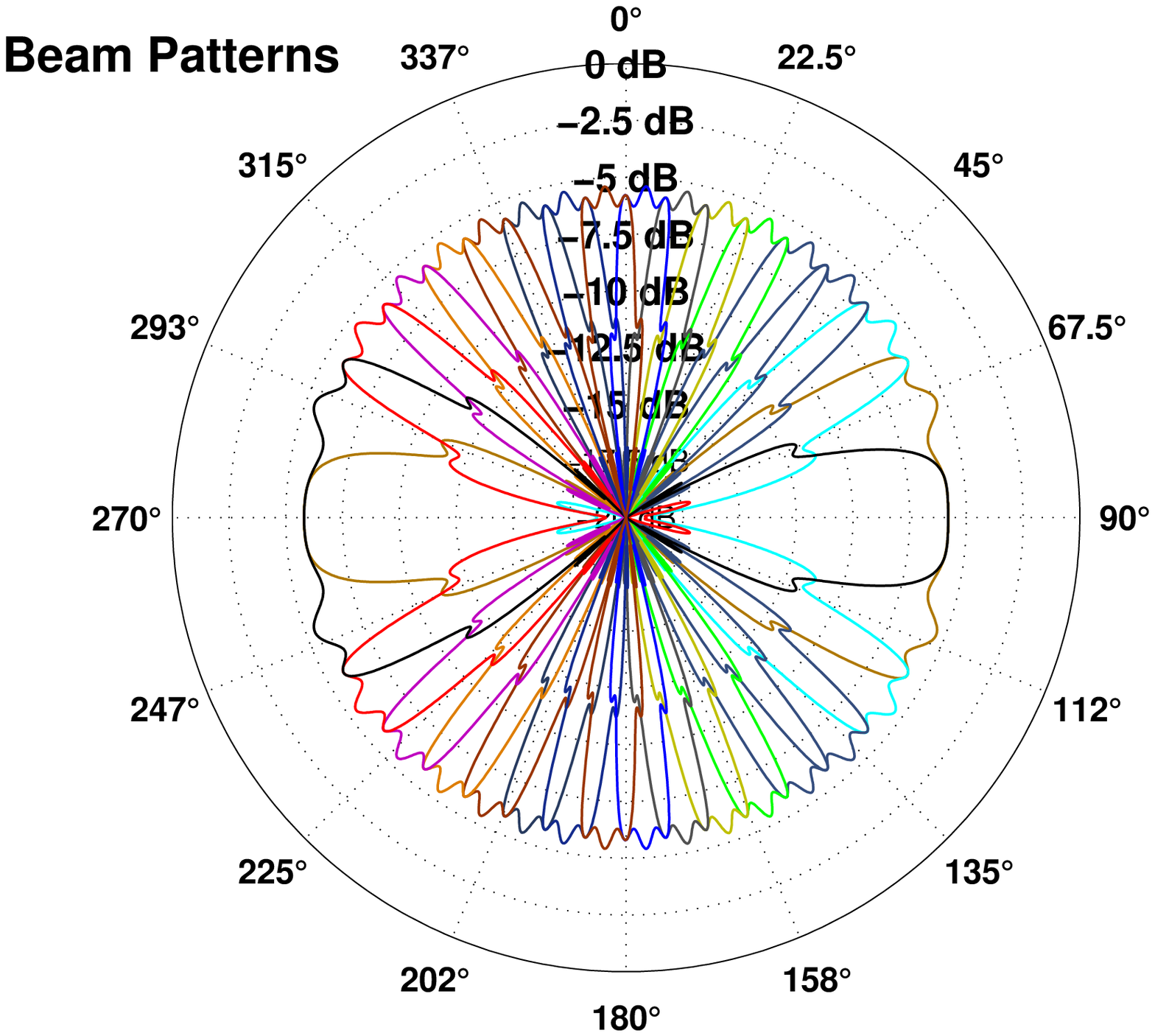}
    \label{fig:subfigA2}
}\vspace{-0.18cm}
\caption{(a) A flattened beam is designed using subarraying and beam spoiling. The blue flattened beam is the superposition of four subarray component beams. (b) Example of subcodebook $\cF_k$ with $N_k=16$ and $q=5$ bits in array size $\Mt=32$. Normalized beam patterns are plotted in polar coordinates. }
\label{fig:CodebookEx2}
\vspace{-0.7cm}
\end{figure}

\section{Simulations and Discussions} \label{SimulationStudy}
We compare the performance of various  beamforming methods using Monte Carlo simulations. 
To evaluate the performance the spatial channel, including angular spread and multipath, is modeled using a street geometry as described in \cite{HongZhang:2010dv}. We consider three multipath components including a line-of-sight (LOS) path and the first order non-line-of-sight reflection from each side of the street. The delays and the angular spread of the reflected paths are calculated using the distance differences and ray-tracing of the geometry model. For the performance comparison, we assumed there was no wind misalignment (i.e., $\overline{u}=0$). The channel is modeled as a Rician channel with K-factor $\mathpzc{K},$ the ratio of the energy in the LOS path to the sum of the energy in other non-LOS paths. For our simulations, the K-factor $\mathpzc{K}$ is set to $13.2$ dB from the channel observation in \cite{MuhiEldeen:2010jl}.

\subsection{Performance Comparison}

Let the average beamforming gain be defined as the SNR gain for a system employing beamformer  $\bff_{opt}$ and  combiner $\bz_{opt}$ as
\vspace{-0.3cm}
\beqn
G_{BF} = E \left[  \left|\bz_{opt}^* \bH \bff_{opt} \right|^2 \right] .
\eeqn
The performance of a non-adaptive joint alignment, a single-sided alignment, and the proposed adaptive sampling alignment were simulated for various settings.
The non-adaptive joint alignment sounds all possible pairs of beamformers and combiners in \eqref{eq:NonAdapSubSampling} using codebooks with sizes given by $\textrm{card}(\cF) =\textrm{card}(\cZ) \approx  \sqrt{L} $. Then, the optimal beamformer $\bff_{opt}$ and combiner $\bz_{opt}$ are chosen according to \eqref{B1.1} and observations in \eqref{eq:NonAdapSubSampling}. The single-sided alignment samples the subspace by searching for $\bff_{opt}$ given a fixed combiner and then searching for $\bz_{opt}$ for a fixed beamformer. The sampling requires  codebooks of sizes $\textrm{card}(\cF) =\textrm{card}(\cZ)= L/2$. After sampling with codebooks, the optimal beamformer $\bff_{opt}$ and the optimal $\bz_{opt}$ are also selected. This alignment approach is currently used in IEEE 802.11ad \cite{IEEE:80211ad}. \textit{We adjust the codebook and subcodebook sizes used in the adaptive sampling such that the total search time $L$ is the same for all the schemes.}

For both the non-adaptive joint search and single-sided search, the transmit and receive array weight codebooks are constructed by quantizing the sets of departure and arrival angles. Both the transmit and receive arrays are assumed to be ULAs. For example, this means the codebook for the transmit beamformer has an $i$th beamformer given by
\vspace{-0.3cm}
\beqn
\bff_i = \frac{1}{\sqrt{\Mt}}\left[ 1 \ \ e^{-j 2 \pi (d/\lambda) \sin(\theta_i)}  \ \ \cdots \ \ e^{-j (\Mt-1) 2\pi (d/\lambda) \sin(\theta_i)}\right]^T
\eeqn
where 
$\theta_i \in \cT$ is the $i$th angle in a uniformly quantized set in \eqref{eq:psi_i}.
The coverage for the sector is specified with  $\cT = \left[-\frac{\pi}{2}, \frac{\pi}{2} \right]$ for a half-wavelength spaced ULA.
Both the beamformer and combiner are implemented on each element using 5 bit phase quantization.
  
\begin{figure}[htbp]
	\centering
		\includegraphics[height=3.0in]{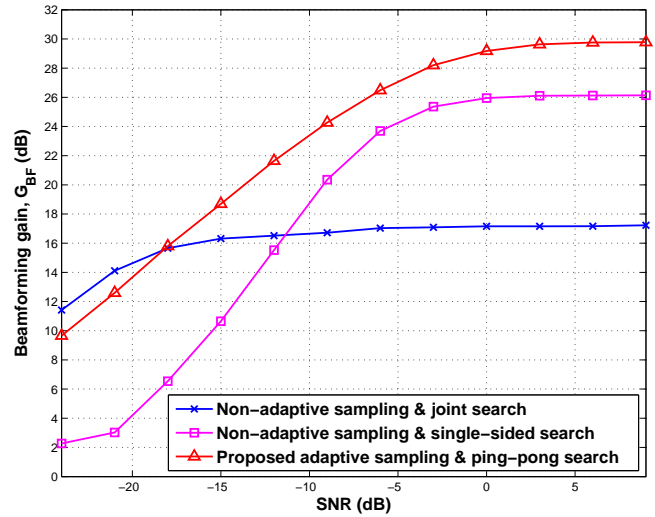}
	\caption{Comparison of beamforming gain versus SNR for $M=32$ array, $L=48$.}
	\label{fig:Results1}
\end{figure}
In the first simulation, with results shown in Fig. \ref{fig:Results1}, the beamforming gain performance of the alignment methods  is shown versus SNR for a fixed value $L = 48.$  The transmitter and receiver both use a ULA of size $M=32.$  For our proposed algorithm,  we design the subcodebooks using ping-pong sounding with  $K=3$ and with $L_K = 8$. The three subcodebooks have sizes, $N_1 = 8$,  $N_2=32,$ and $N_3 = 64.$
The non-adaptive sampling codebook has seven vectors, and the single-sided sampling codebook is of size 24.  
At high SNR,  the proposed adaptive alignment method has around a $4$ dB gain improvement over single-sided alignment and more than a $13 \ \textrm{dB}$ gain over the non-adaptive joint alignment. From this result, the operating SNR is set to $5$ dB for all other simulations.
 
\begin{figure}[htbp]
	\vspace{-0.3cm}
	\centering
		\includegraphics[height=3.0in]{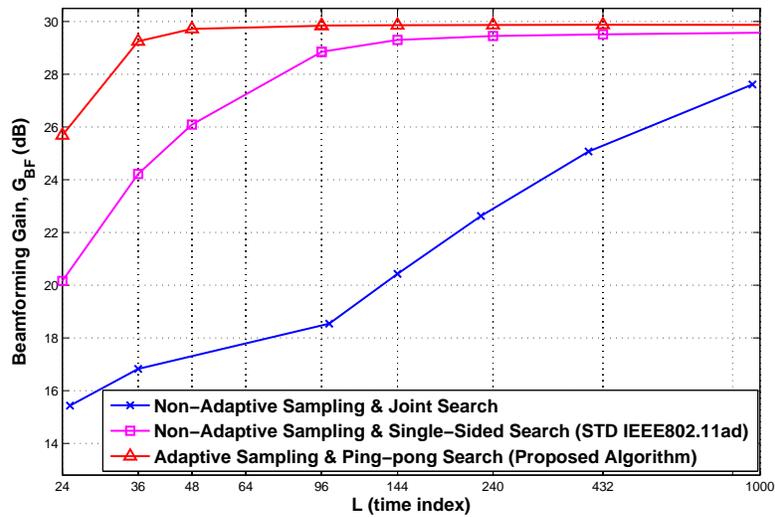}
	\caption{Comparison of beamforming gain versus beam search time $L$ for $M=32$ array, SNR = $5$ dB.}
	\label{fig:Results2}\vspace{-0.7cm}
\end{figure}
In order to investigate the effect of the size of subspace samples, Fig. \ref{fig:Results2} compares the beamforming gain versus the search time $L$. 
Also, the codebook and subcodebook sizes are a function of the search time $L.$
Non-adaptive joint sampling and non-adaptive single-sided sampling choose the codebook sizes as $\sqrt{L}$ and $L/2$, respectively. In adaptive sampling, the number of sounding vectors at each level in the hierarchy is set as a function of $L$
due to the relationship between $L$ and $L_K$ with $K$, i.e., $L_K = \frac{L}{2K}$.
For example, two-level codebooks (i.e., $K=2$) with $L_K=6$ are used for the small search time $L=24$ because $L$ is not large enough to sound in multiple levels. In $L=36$ and $L=48$, three-level sounding and sampling $K=3$ is performed with $L_K=6$ and $L_K=8$, respectively. For $L$ greater than $96$ as shown in Fig. \ref{fig:Results2},
subcodebooks with levels $K=4$ or greater are used for adaptive sampling and sounding with the proper $L_K$.
The transmitter and receiver both use ULAs of size $M=32,$ and the SNR is fixed at $5$ dB. 
Note that the beamforming gain reaches its bound $G_{BF,max} = 10\log_{10}(\Mt \Mr)$ as $L$ grows large. 
When the search time $L$ grows large, 
the subcodebook size of the last level, $N_K$ (with $N_K\leq(L_K)^K$), is much bigger than the array size $M$, and the beamforming gain is bounded by the full CSI beamforming gain.
To compare performance, consider a target beamforming gain of $26$ dB. To  achieve this gain,  the adaptive alignment algorithm only requires $L \approx 25$.  In contrast,  single-sided alignment requires $L \approx 47,$ and non-adaptive joint alignment requires $L \approx 585$.
The adaptive sampling and alignment method using the hierarchy relationship in \eqref{eq:SubcodebookSelection} efficiently estimates the optimal beamformer-combiner pairs.
The proposed adaptive sampling and beam alignment scheme also allows the alignment to be accomplished with a smaller search time $L$ than the alignment method used in IEEE 802.11ad \cite{IEEE:80211ad}. This advantage of the proposed technique is even more significant for larger arrays, which are simultaneously more susceptible to wind sway and require a larger codebook size with $L \ll \Mr\Mt$.

\begin{figure}[htbp]
	\centering
		\includegraphics[height=3.0in]{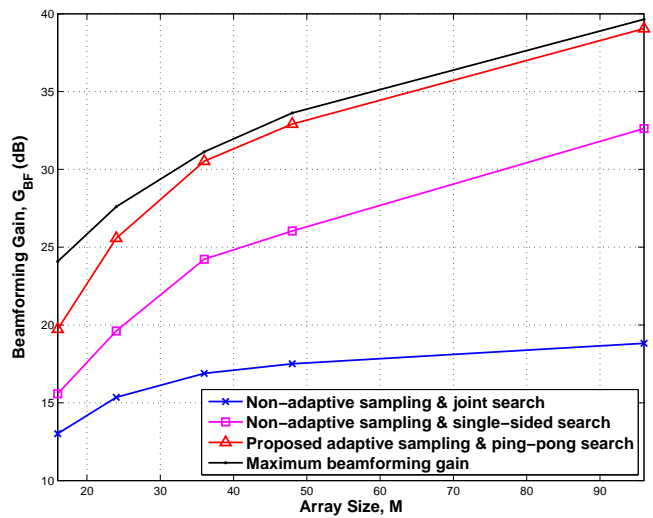}\vspace{-0.5cm}
	\caption{Comparison of beamforming gain versus array size $M=\Mt=\Mr$ for $L=M$, SNR = $5$ dB.}
	\label{fig:Results3}\vspace{-0.7cm}
\end{figure}
The size of the array also plays a major role in beamforming gain. Fig. \ref{fig:Results3} demonstrates the beamforming gain as a function of the transmit and receiver array size $M.$ In this plot, the SNR is  $5$ dB, and $L$ varies with the array size according to $L = M$. 
The subcodebooks for $M$ are designed with $K=3$, except the small array sizes $M=16$ and $M=24$ with $K=2$.
 The aligned beamforming gain is compared with the theoretical array gain, which is   $10 \log_{10}(\Mt\Mr)$ at perfect alignment. From the simulation results it can be seen that there is a consistent performance gap between the adaptive alignment and the single-sided alignment. The adaptive alignment scheme outperforms the non-adaptive joint alignment for all array sizes. 
In Fig. \ref{fig:LinkGain}, a target beamforming gain is roughly calculated as $29$ dB at $100$ m distance. Utilizing the adaptive alignment, the system with $M=32$ is able to achieve the same beamforming gain as a much larger array system with $M \approx 70$ using the single-sided alignment.

Fig. \ref{fig:Results4} compares the beamforming gain versus the average wind speed with different array sizes. 
The beamforming gain $G_{BF}$ is averaged over different wind-environment realizations.
The figure demonstrates that increasing the array size in only one dimension (e.g., in a uniform linear array) comes with a possibly severe penalty from wind sway. In the case of $M=96$, the performance degradation due to the wind sway misalignment is up to 10 dB at 40 m/s wind speed.
\begin{figure}[htbp]
	\vspace{-0.7cm}
	\centering
		\includegraphics[height=2.8in]{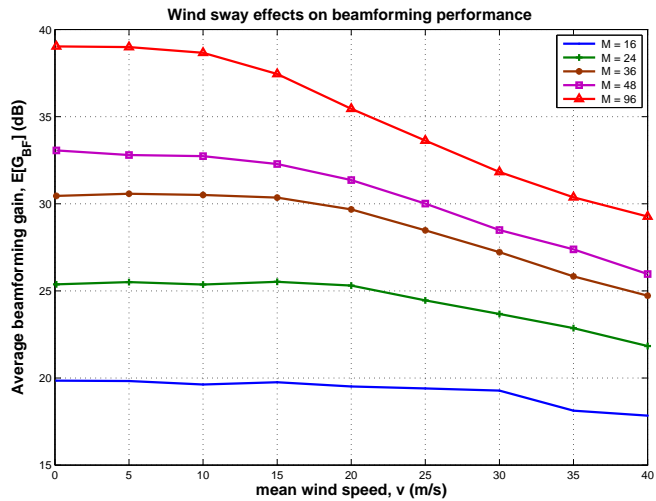}\vspace{-0.5cm}
	\caption{Average beamforming gain with wind sway effect.}
	\label{fig:Results4}
\end{figure}

\vspace{-0.8cm}
\section{Conclusions} \label{Conclusions}
\vspace{-0.3cm}
In this paper, we  studied the use of millimeter wave wireless communication for both backhaul and access in small cell networks.  The longer communication link distances, combined with the severe path loss at millimeter wave frequencies, make aligning the transmit and receive beams a challenging and important problem.  
We addressed the problem of subspace sampling, which provides the observation data needed for the beam alignment algorithms.  Subspace sampling can be done in a non-adaptive or adaptive manner.  Adaptive sampling techniques can leverage previous received data to substantially improve system performance.  Simulations compared our various proposed beam alignment and subspace sampling algorithms.   Because outdoor picocell access points will most likely be mounted to poles, we also discussed the problem of pole sway due to wind and the tradeoff of the array size and achievable receive SNR. We modeled wind-induced impairments in the millimeter wave beamforming system and evaluated their effect on the beamforming gain.

This paper documented the tradeoff between array size and wind-induced movement.  More work is needed in this area.  When the required number of transmit and receive antenna elements is large, it appears to be insufficient to employ large uniform linear arrays.  A more resilient architecture is to use a two-dimensional uniform linear array, which limits the effect of beam misalignment by spacing the array over two dimensions.  

\vspace{-0.5cm}

\appendices
\section{Wind Vibration Modeling}\label{WindModel}
The wind vibration modeling in this section follows the approach in \cite{Piersol:2009,simiu:1996wind,Davenport:1961bb} as illustrated in Fig. \ref{fig:WindModel} showing the relation between wind velocity, drag, and pole response. 
\begin{figure}[b!]
	\centering
	\includegraphics[height=2in]{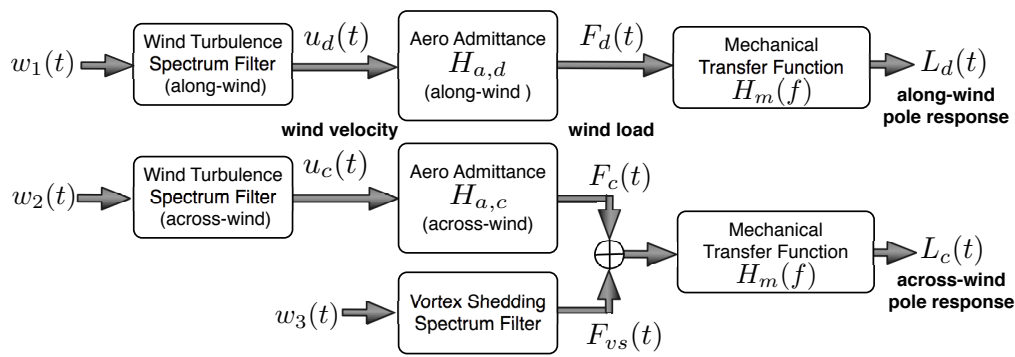}
	\vspace{-0.2cm}
	\caption{Block diagram representation of transfer function mapping wind velocity to pole response.}
	\label{fig:WindModel}
	\vspace{-0.8cm}
\end{figure}
In general, the wind velocity is a time-varying, vector-valued spatial random field $\bu(x,y,z,t)$ where $x$ and $y$ represent the surface dimension variables and $z$ represents height above the surface. Let $\overline{\bu}$ represent the mean wind vector as averaged over a suitable time window. Though $\overline{\bu}$ is a function of the spatial variables we assume: 
1) it is approximately constant over small variations in the surface variables $x$ and $y$ (since we consider only one short link whose distance is on the order of 100 m, 
2) that the height variable $z$ is fixed and equal to 10 m (a reasonable choice for the height of a pole to which the millimeter wave antenna array is mounted), and 
3) that the $z$ component of $\overline{\bu}$ is equal to zero. 
Then following the common practice \cite{Piersol:2009,simiu:1996wind,Davenport:1961bb} we suppress the spatial variables to simplify notation and write the wind velocity as
\vspace{-0.5cm}
\beq
\bu(t) = \overline{\bu} + \bd u_d(t) + \bc u_c(t) \notag
\eeq
where $\bd$ is a unit vector in the direction of the mean wind and $\bc$ is a unit vector orthogonal to the direction of the mean wind. The time-varying components $u_d(t)$ and $u_c(t)$ represent turbulence in the along-wind and across-wind directions, respectively. It is also assumed that there is no turbulence component in the $z$ direction, which is a reasonable assumption near the surface ($z$ = 10 m). Furthermore, the top of the pole to which the antenna array is mounted is approximately constrained to move only in the $(x,y)$ plane.

The turbulence components $u_d(t)$ and $u_c(t)$ are modeled as zero-mean, wide-sense stationary, uncorrelated random processes with power spectral densities given by \cite{simiu:1996wind}
\vspace{-0.25cm}
\[
S_{u_d}(f) = \frac{500 u_*^2}{\pi \overline{u}} \left[ \frac{1}{1+ 500 f /2\pi \overline{u}} \right]^{5/3}, \label{eq:S_longitudinal}
\]
\vspace{-0.5cm}
\[
S_{u_c}(f) = \frac{75 u_*^2}{2\pi \overline{u}} \left[ \frac{1}{1+ 95 f /2\pi \overline{u}} \right]^{5/3} \label{eq:S_lateral}
\]
where $f$ is frequency in Hz, $\overline{u} = \| \overline{\bu} \|$, and $u_* = {\overline{u}}/{ 2.5 \ln\left( 10/{z_0}  \right) }$ is the shear velocity at height $z = 10$ m, obtained from the terrain roughness length parameter $z_0$, which has been characterized in \cite{Wieringa:1986roughness}. For an urban scenario $z_0 = 2$ m. In Fig. \ref{fig:WindModel} the turbulence components are modeled as the outputs of wind spectrum filters when driven by independent, white Gaussian processes.

The wind exerts force on the antenna mounting pole via two mechanisms. The first, and most important, is through drag in response to the mean wind speed and the two orthogonally directed wind speed turbulence components. This effect is modeled in Fig. \ref{fig:WindModel} using the two aero-admittance functions $H_{a,d}$ and $H_{a,c}$. Although the mean wind term $\overline{u}$ exerts a significant drag force we can and do ignore it because of our assumption that the poles on either end of a short link are exposed to the same mean wind components. Therefore, mean wind will cause no relative displacement between the poles and hence no corresponding pointing error. Despite this the magnitude of the mean wind will still be an important factor because the power in the random processes $u_d(t)$ and $u_c(t)$ is proportional to $\overline{u}^2$. The second method by which wind exerts force on a pole is via a phenomenon known as vortex shedding which occurs for aerodynamically shaped bodies such as poles with a circular cross section. The mean wind passing by the pole causes a lifting force in the across-wind direction which is approximately periodic with a period equal to the time between the shedding of vortices from a common side.

The time-varying drag forces due to the orthogonal wind speed turbulence components are given by zero-mean random processes $F_d(t)$ and $F_c(t)$ with power spectral densities \cite{simiu:1996wind}
\vspace{-0.3cm}
\beqn
S_{F_d}(f) &=& \left|H_{a,d}\right|^2 S_{u_d}(f) = \left(2 \kappa \overline{u}\right)^2 S_{u_d}(f) , \\
S_{F_c}(f) &=& \left|H_{a,c}\right|^2 S_{u_c}(f) = \left(\kappa \overline{u}\right)^2 S_{u_c}(f)
\eeqn
where $\kappa = \frac{1}{2} \rho_a C_D A_e$, $\rho_a$ is the density of air in $\mbox{kg/m}^3$, $C_D$ is the coefficient of drag, and $A_e$ is the effective area in $\mbox{m}^2$.

The vortex shedding component of the force also acts in the across-wind direction and so it is added to the across-wind force due to drag as shown in Fig. \ref{fig:WindModel}. Vortex shedding is characterized by the vortex-shedding frequency $f_{vs}$ given by \cite{Yam:1997vx,Caracoglia:2007}
\vspace{-0.2cm}
\beq
f_{vs} &=& \cS\frac{\overline{u}}{d_p} \label{eq:VortesShedding}
\eeq
where $\cS$ is the Strouhal number, a constant dependent on the shape of the body, and $d_p$ is the diameter of the pole. For a pole of circular cross-section \cite{Yam:1997vx} suggests $\cS=0.2$. The power spectral density of the across-wind force due to vortex-shedding can be written
\vspace{-0.1cm}
\beq
S_{F_{vs}}(f) = \kappa^2 \frac{1.125}{\sqrt{\pi} f f_{vs}} \exp \left( - \left[ \frac{1 - f/f_{vs}}{0.18} \right]^2 \right) .\vspace{0.1cm}
\eeq
Details and certain parameter choices are further discussed in \cite{simiu:1996wind}.
The vortex shedding frequency is calculated as $f_{vs} = 12$ Hz with  $d_p=50$ cm in \eqref{eq:VortesShedding}. 

The various component forces then drive a mechanical model for the antenna-mounting pole. We assume that the along wind and across wind forces are independent and that there is no coupling between pole dynamics in the two directions. The pole is modeled as a simple spring-mass-damper system \cite{simiu:1996wind} characterized by a damping coefficient and a natural frequency, $\zeta$ and $f_n$ in Hz, respectively. The mechanical transfer function is given by
\vspace{-0.4cm}
\begin{equation}
H_m(f) = \frac{1}{ 4 m \pi^2 f_n^2 \left( \left[1 - (f/f_n)^2\right]^2 + 4 \zeta^2 (f/f_n)^2 \right)^{1/2} } \vspace{0.1cm}\nonumber
\end{equation}
where $m$ is the mass of the light and pico-cell antenna on the top of the pole. Pole responses in the along-wind and across-wind directions are computed as the outputs of mechanical transfer function systems when driven by wind force random processes as illustrated in Fig. \ref{fig:WindModel}. Finally, the simulation is computed in the spectral domain from the pole-displacement power spectral densities (e.g., $S_{L_d}(f) = |H_m(t)|^2 S_{F_d}(f)$, $S_{L_c}(f) = |H_m(t)|^2 S_{F_c}(f)$), using the spectral representation method and the inverse FFT as in \cite{shinozuka:1991simulation}. The simulation can represent a pole-response bandwidth up to 10 Hz.

\bibliographystyle{IEEEtran}
\scriptsize{\bibliography{IEEEabrv,Multilevel_codebook_JNL_Wind_rev}}

\end{document}